\documentclass[amsmath,nobibnotes,aps,superscriptadress,amssymb,pra,aps,showpacs,superscriptaddress,twocolumn,longbibliography,preprintnumbers]{revtex4-1}
\usepackage{amsmath,amsfonts,amssymb,amsthm,graphicx,epsfig,bbm}
\usepackage[colorlinks=true,citecolor=blue,linkcolor=blue,urlcolor=blue]{hyperref}
\usepackage[usenames]{color}
\usepackage[dvipsnames]{xcolor}
\usepackage{tikz}
\usepackage{float}
\usepackage{amsmath}
\usepackage{dcolumn}
\usepackage{multirow,array,booktabs}
 
\usepackage{bm}
\usepackage{color}
\usepackage{tabularx}
\usepackage{times}
\usepackage{epstopdf}
\usepackage{amscd}
\usepackage{amsthm}
\usepackage{amssymb}
\usepackage{bbm}
\usepackage{amssymb}
\usepackage{amstext}
\usepackage{latexsym}
\usepackage{hyperref}
\usepackage{subfigure}
\usepackage{psfrag}
\usepackage{soul,xcolor}
\usepackage[normalem]{ulem}
\usepackage{dsfont}
\usepackage{txfonts}
\usepackage{physics}
\usepackage{footnote}
\usepackage{multirow}
\usepackage{appendix}
\usepackage{mathtools}
\usepackage{tikz}
\usetikzlibrary{quantikz}
\usepackage{bbold}

\usepackage{enumitem}
\usepackage{xstring}
\usepackage{cancel}
\usepackage{amscd}
\usepackage{bbm}

\usepackage{units}
\usepackage{cancel}
\usepackage{xstring}
\usepackage{import}
\usepackage{subfiles}
\usepackage{amsmath,amsfonts,amssymb,amsthm,graphics,graphicx,epsfig,bbm}
\usepackage[colorlinks=true,citecolor=blue,linkcolor=blue,urlcolor=blue]{hyperref}
\usepackage[usenames]{color}
\usepackage{amsmath}
\usepackage{epsfig}
\usepackage{dcolumn}
\usepackage{bm}
\usepackage{color}
\usepackage{times}
\usepackage{epstopdf}
\usepackage{amssymb}
\usepackage{amstext}
\usepackage{latexsym}
\usepackage{hyperref}
\usepackage{amsfonts}
\usepackage{psfrag}
\usepackage{soul,xcolor}

\usepackage[normalem]{ulem}
\usepackage{dsfont}
\usepackage{txfonts}
\usepackage{physics}
\usepackage{xcolor}
\usepackage{footnote}
\usepackage{multirow}
\usepackage{appendix}
\usepackage{mathtools}
\usepackage{ulem}
\usepackage{mathtools}

\definecolor{mycolor}{RGB}{106,81,162}

\newtheorem{remark}{Remark}

\usepackage{mathtools}

\newcommand{\iden}[1]{
    \ifthenelse{\equal{1}{\string #1}}
  {
   \mathbbm{1}
  }
  {
   \mathbbm{1}^{\otimes#1}}
  }
\newcommand{\ketzero}[1]{
    \ifthenelse{\equal{1}{\string #1}}
  {
   \ket{0}
  }
  {
   \ket{0}^{\otimes#1}}
  }
\newcommand{\brazero}[1]{
    \ifthenelse{\equal{1}{\string #1}}
  {
   \bra{0}
  }
  {
   \bra{0}^{\otimes#1}}
  }
\newcommand{\ketone}[1]{
      \ifthenelse{\equal{1}{\string #1}}
    {
     \ket{1}
    }
    {
     \ket{1}^{\otimes#1}}
    }
  \newcommand{\braone}[1]{
      \ifthenelse{\equal{1}{\string #1}}
    {
     \bra{1}
    }
    {
     \bra{1}^{\otimes#1}}
    }

\usepackage{physics}  
\usepackage{dsfont}

\begin{document}

\author{Nikita Guseynov}
\email{guseynov.nm@gmail.com}
\affiliation{Global College, Shanghai Jiao Tong University, Shanghai 200240, China.}

\author{Mikel Sanz}

\affiliation{Department of Physical Chemistry, University of the Basque Country UPV/EHU, Apartado 644, 48080 Bilbao, Spain}
\affiliation{EHU Quantum Center, University of the Basque Country UPV/EHU, Apartado 644, 48080 Bilbao, Spain}
\affiliation{IKERBASQUE, Basque Foundation for Science, Plaza Euskadi 5, 48009, Bilbao, Spain}
\affiliation{Basque Center for Applied Mathematics (BCAM), Alameda de Mazarredo, 14, 48009 Bilbao, Spain}

\author{\'Angel Rodr\'iguez-Rozas}
\affiliation{Corporate \& Investment Banking, Banco Santander, Avenida de Cantabria S/N, 28660 Boadilla del Monte, Madrid, Spain}

\author{Nana Liu}
\affiliation{Global College, Shanghai Jiao Tong University, Shanghai 200240, China.}
\affiliation{Institute of Natural Sciences, School of Mathematical Sciences, Shanghai Jiao Tong University, Shanghai 200240, China}
\affiliation{Ministry of Education Key Laboratory in Scientific and Engineering Computing, Shanghai Jiao Tong University, Shanghai 200240, China}

\author{Javier Gonzalez-Conde}
\email{javier.gonzalezc@ehu.eus}
\affiliation{Department of Physical Chemistry, University of the Basque Country UPV/EHU, Apartado 644, 48080 Bilbao, Spain}
\affiliation{EHU Quantum Center, University of the Basque Country UPV/EHU, Apartado 644, 48080 Bilbao, Spain}
\affiliation{Quantum Mads, Calle Larrauri 1, Edificio A, piso 3, puerta 28, Derio, Spain}

\title{Quantum Algorithm for Local-Volatility Option Pricing via the Kolmogorov Equation}
\date{\today}

\begin{abstract}
The solution of \textit{option-pricing problems} may turn out to be computationally demanding due to non-linear and path-dependent payoffs, the high dimensionality arising from multiple underlying assets, and sophisticated models of price dynamics. In this context, quantum computing has been proposed as a means to address these challenges efficiently. Prevailing approaches either simulate the stochastic differential equations governing the forward dynamics of underlying asset prices or directly solve the backward pricing partial differential equation. Here, we present an \textit{end-to-end} quantum  algorithmic framework that solves the Kolmogorov forward (Fokker-Planck) partial differential equation for local-volatility models by mapping it to a Hamiltonian-simulation problem via the Schr\"odingerisation technique. The algorithm specifies how to prepare the initial quantum state, perform Hamiltonian simulation, and how to efficiently recover the option price via a swap test. In particular, the efficiency of the final solution recovery is an important advantage of solving the forward versus the backward partial differential equation. Thus, our end-to-end framework offers a potential route toward quantum advantage for challenging option-pricing tasks. In particular, we obtain a polynomial advantage in grid size for the discretization of a single dimension. Nevertheless, the true power of our methodology lies in pricing high-dimensional systems, such as baskets of options, because the quantum framework admits an exponential speedup with respect to dimension, overcoming the classical curse of dimensionality.

\end{abstract}

\maketitle
\newtheorem{theorem}{Theorem}[section]

\section{Introduction}

Option-pricing models are mathematical frameworks used in financial markets to estimate the fair value of options---contracts that grant the right, but not the obligation, to buy or sell an underlying asset at a predetermined strike price \(K\) on or before a specified maturity \(T\). These models underpin valuation, hedging, and risk management for investors and financial institutions.

  Solving option-pricing models can be computationally demanding for several reasons.  A primary driver is the complexity of option payoffs and their dependence on multiple market factors—such as underlying asset prices, interest rates, volatilities, and time to maturity—which can introduce nonlinearity and path dependence \cite{björk2004arbitrage,Asian_options_clasical,Exotic_options_clasical,Bermudan_options_clasical,American_options_clasical,Barrier_options_clasical}. In addition, many financial models involve multiple state variables and risk factors, leading to high-dimensional problems that suffer from the curse of dimensionality \cite{clewlow1998implementing,Glas04,hout2021numericalvaluationamericanbasket,10.1007/978-3-642-31703-3_6,doi:10.1137/09077271X}. For example, pricing options on multiple underlying assets or incorporating multiple sources of risk can significantly increase computational complexity—even when the underlying price dynamics are analytically solvable \cite{Multioption_classical_1,Multioption_classical_2,Multioption_classical_3}. Moreover, the dynamics of many sophisticated pricing models do not admit closed-form solutions and instead require numerical methods to obtain approximate solutions \cite{dupire1994pricing,Heston_classical,Heston_classical_2,SABR_classical_model,SLV}.

\begin{figure}[b!]
    \includegraphics[width=0.35\textwidth]{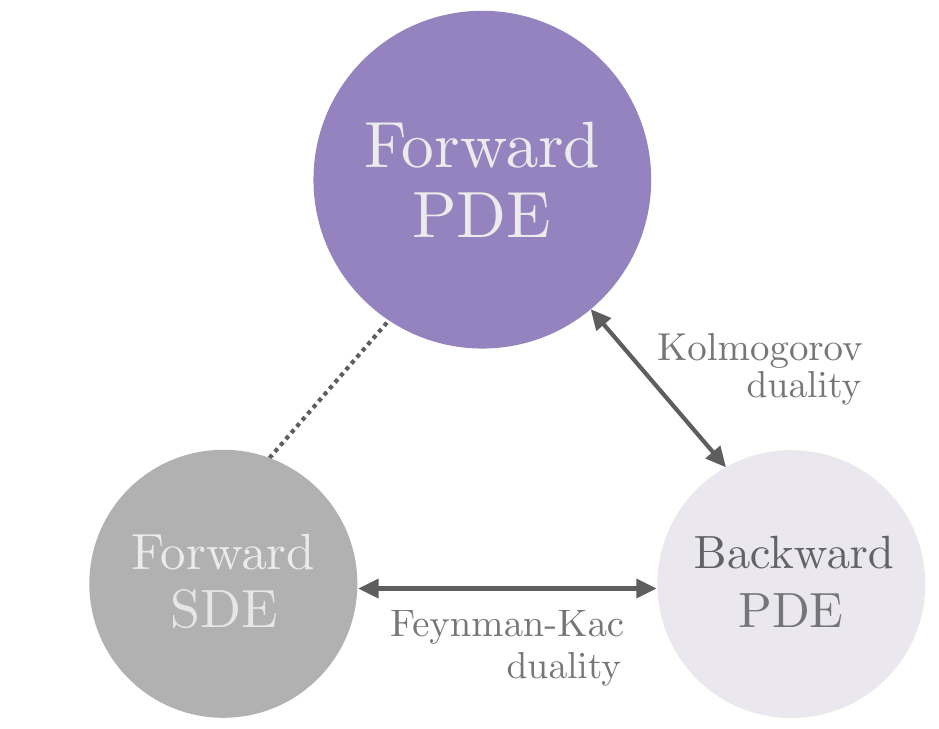}
    \caption{Equivalent differential-equation formulations for option-pricing. Forward models simulate the dynamics of the underlying asset price from the current date forward, while backward models propagate the option value from maturity back to today. In the forward approach, the option price is computed as the expected payoff at maturity, discounted at the risk-free interest rate. We can also distinguish between PDE formulations—the Kolmogorov equations \cite{Conze2008TheFK}—and SDE formulations (e.g., geometric Brownian motion), connected to the backward PDE through the Feynman-Kac formula \cite{black1973pricing}.
}
    \label{fig:figure_1}
\end{figure}

\begin{figure*}
    \includegraphics[width=1\textwidth]{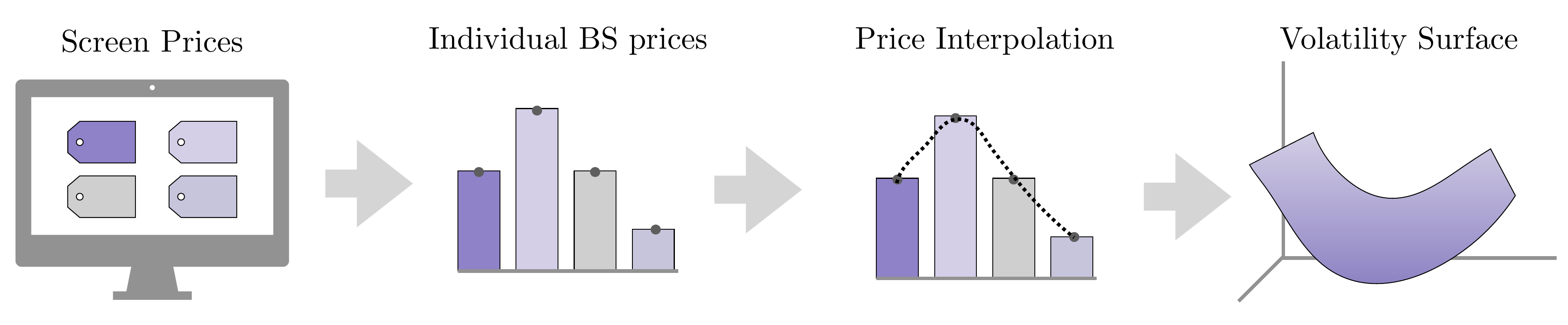}
    \caption{
   A unique diffusion process consistent with risk-neutral densities implied by European option prices is represented by the local-volatility (LV) model. This process is constructed as follows: (i) observe implied volatilities across strikes \(K\); (ii) map each implied volatility to an option price via the Black--Scholes formula, yielding discrete price points (each associated with a specific volatility); (iii) interpolate these points to obtain a continuous option-price surface over strikes and maturities; and (iv) compute Dupire’s local-volatility \(\sigma(S_{\tau},\tau)\) from this continuous surface.
}
    \label{fig:figure_local_vol}
\end{figure*}

These numerical methods can be computationally intensive—particularly for high-dimensional problems or when high precision is required—which often leads, in practice, to solving simplified versions of the models. Finding efficient classical algorithms that yield accurate solutions for complex models remains an active area of research in computational finance \cite{beliaeva2010simple,hout2010adi,alos2012decomposition,CHIARELLA20122034,he2018closed}. Consequently, these challenges can push classical computational resources to their limits, sometimes rendering them insufficient for accurate pricing within a reasonable time frame.

 In this context, quantum computing emerges as a promising avenue for developing more efficient computational techniques to address the aforementioned challenges in option-pricing \cite{Martin_gale, bermudan_option_Miyamoto, Rebentrost_stopping, koppe_asian_Valuation_Trees,  cibrario_rainbow, Latorre_unary, Stamatopoulos_QSP_payoff, Kan_ORCA_Heston, Miyamoto_smile,ORUS_Tensor_neural_network_Heston, kubo_solving_SDE, Montanaro_multilevel_MC_local_vol, Neufeld_CPWA, woerner_option_pricing, Woerner_threshold, Miyamoto_and_Kubo_finite_difference, Miyamoto_varitional,  Woerner_heston, Miyamoto_TN_series, Kan_ORCA_Heston, Nield_S&P, ORUS_Tensor_neural_network_Barenblatt,Montanaro_multilevel_MC_local_vol, Miyamoto_MC_pseudo-random-number_generator,Herbert_MCI, Fogarasi_MC, Merecek_survey_MC, gonzalez2021pricing, Ana_PCA, Woerner_greeks, Jin_LIBOR, Woerner_fk, Review_Yndurain, Review_wilkens, Review_intallura, Review_gomez_Cesga, ORUS2019100028, Review_herman2022, Pracht, kumar2024simulatingnonhermitiandynamicsfinancial}. 
Quantum algorithms have the potential to significantly accelerate computations for problems arising in option-pricing models, including cases with complex payoffs \cite{Martin_gale,bermudan_option_Miyamoto,Rebentrost_stopping,koppe_asian_Valuation_Trees,cibrario_rainbow,Latorre_unary,Stamatopoulos_QSP_payoff,Kan_ORCA_Heston,Miyamoto_smile,ORUS_Tensor_neural_network_Heston,Montanaro_multilevel_MC_local_vol,Neufeld_CPWA,woerner_option_pricing,Woerner_threshold,Miyamoto_and_Kubo_finite_difference}, multi-asset options \cite{Neufeld_CPWA,woerner_option_pricing,Woerner_threshold,Miyamoto_and_Kubo_finite_difference,Miyamoto_varitional}, and complex price dynamics \cite{ORUS_Tensor_neural_network_Barenblatt,ORUS_Tensor_neural_network_Heston,kubo_solving_SDE,Miyamoto_smile,Woerner_heston,Miyamoto_TN_series,Kan_ORCA_Heston,Nield_S&P,Montanaro_multilevel_MC_local_vol}. In particular, quantum amplitude estimation provides a quadratic speedup in query complexity for Monte Carlo integration relative to classical sampling methods \cite{Nield_S&P,Miyamoto_MC_pseudo-random-number_generator,Herbert_MCI,Fogarasi_MC,Merecek_survey_MC}. Harnessing these capabilities could enable more sophisticated pricing models that better handle the high-dimensional and computationally demanding nature of option-pricing simulations.

A closely related work \cite{Woerner_fk} studies the Feynman--Kac equation through a variational-style propagation that evolves solutions of $\partial_t u=\mathcal{L}u - Vu$, working explicitly with the Kolmogorov forward form. The method applies McLachlan’s principle to compute the time update for the ansatz parameters directly, so no outer optimization loop is required. Each time step advances the ansatz $|\psi(\vec{\theta})\rangle$ by iteratively solving for a parameter increment $\delta\vec{\theta}$. In particular, the update computes the change between $V(t)$ and $V(t+\mathrm{d}t)$ by determining the corresponding $\delta\vec{\theta}$. This design helps avoid some gradient pathologies and exhibits practical stability over short horizons. However, theoretical guarantees remain limited—general convergence rates and error bounds have not yet been established. The reported cost is $N_\tau(N_\theta^3)$ per time step, with $N_\tau$ the number of steps and $N_\theta$ the number of ansatz parameters. What remains open is an analysis that ties $N_\tau$ and $N_\theta$ to model characteristics—such as local-volatility $\sigma(x,t)$, payoff discontinuities, or stiffness—to clarify resource scaling and guide design choices.

In this work, we analyze and demonstrate the potential of quantum computers to solve the Kolmogorov forward equation under a local-volatility (LV) model—a non-constant volatility setting that enables more realistic financial modeling. Our contribution extends prior quantum approaches to option-pricing \cite{gonzalez2021pricing} by adopting a forward partial differential equation (PDE) formulation in place of the standard backward PDE (both equivalent), see Fig.~\ref{fig:figure_1}.  Relative to solving the stochastic differential equation (SDE) for the underlying price dynamics, the Kolmogorov forward PDE is often numerically advantageous in computational finance, for example in path-dependent options, where solving the pricing problem backward using a numerical scheme can provide higher accuracy than least-squares Monte Carlo \cite{glasserman2004monte}. We also investigate the forward–backward duality in quantum implementations of option-pricing algorithms.

Our manuscript presents a methodology within the quantum input--processing--output (IPO) framework for solving the Kolmogorov forward equation by mapping it to a Hamiltonian-simulation problem via the Schr\"odingerisation technique \cite{Jin_2024, analog,jin2022quantum, jin2025schrodingerizationmethodlinearnonunitary}. Thus, it can be considered an \textit{end-to-end} quantum algorithm that takes classical input data and outputs classical information. We provide an algorithm that specifies how to (i) encode the initial classical data into a quantum state, (ii) carry out the Hamiltonian simulation, and (iii) efficiently recover the option price. Without further refinement in studying effective regimes, our methodology  for a single option-pricing instance via Schr\"odingerisation provides a polynomial advantage in the number of qubits. However, the technique’s true potential emerges when pricing baskets of options, where it can yield exponential speedups and mitigate the curse of dimensionality.

\section{option-pricing models for European vanilla options}

Option-pricing refers to the process of determining the fair value of a financial derivative contract under the risk-neutral framework \cite{hull2018options}. In equity, European-type options are derivative contracts that grant the holder the right, but not the obligation, to buy (call option) or sell (put option) an underlying stock asset at a specified strike price at a predetermined date, specifying the time maturity of the deal. The price of a derivative depends on the inherent randomness stemming from the time evolution of the underlying stock and interest rates, with the latter sometimes assumed to be deterministic, as we do in this manuscript. 
In mathematical finance, the asset $S_{\tau}$ underlying a financial derivative is typically modeled by assuming that it follows a SDE under the risk-neutral measure, of the form

\begin{eqnarray}
\begin{gathered}
       dS_{\tau} = (r_{\tau}-d_{\tau}) S_{\tau}\,d\tau + \sigma_{\tau} S_{\tau}dW_{\tau}\\
        S_{\tau=0}=S_0,
    \end{gathered}
    \label{eq:sto_process}
\end{eqnarray}
where $r_{\tau}$ is the instantaneous risk-free interest rate, $d_{\tau}$ is the dividends rate, such as $\mu_{\tau}=r_{\tau}-d_{\tau}$ gives an average local direction to the dynamics, and $W_{\tau}$ is a Wiener process \cite{Conze2008TheFK, Lawler2005ConformallyIP}, representing the inflow of randomness into the dynamics and governed by the volatility 
$\sigma_{\tau}$. In the simplest scenario, i.e. the Black–Scholes model \cite{black1973pricing}, $\sigma_{\tau}$ is assumed to be constant. However, in reality, the realized volatility varies with time and with the price of underlying in an stochastic way. Thus, in order to achieve a more realistic and accurate description, more complex models have been proposed. In particular, in this article we focus on the LV model, in which the volatility of an asset is assumed to vary with both the asset's price and time, allowing for a more accurate fit to market prices of options across different strike prices and maturities \cite{dupire1994pricing}.  

The stochastic dynamics given by Eq.~(\ref{eq:sto_process})  induces a price probability distribution at time $\tau$, denoted by $\mathbb{P}$, which determines the 
 fair value of the option price at spot $\tau = 0$ (present time) as  
 
\begin{equation}
\label{eq:payoff_definition}
V_0=e^{-\int_0^T r_d(s) ds} \underbrace{E^{\mathbb{P}}_{\tau} \left[ C_0(S_{\tau=T}) | S_{\tau=0} = s_0,  \right]}_{C(S_0, 0)},
\end{equation} where $r_d(\cdot)$ is the domestic risk-free rate that determines the discount factor $\exp(-\int_0^T r_d(s) ds)$, and $C_0(s)$ is the \textit{payoff} of the contract, see some examples in Table \ref{tab:payoff}.

\begin{table}[h!]
    \centering
    \begin{tabular}{c||c}
         Put & ${\large C^\text{put}_0(s)=\text{max}(0,s-K)}$ \\ \hline
         Call & ${\large C^\text{call}_0(s)=\text{max}(0,K-s)}$
    \end{tabular}
    \caption{\textit{Payoff} function for the Put and Call option contracts. Here $K$ denotes the strike price.}
    \label{tab:payoff}
\end{table}

\subsection{Local-volatility (LV) model}
\label{secion: local_vol}
The LV model treats volatility, $\sigma(S_{\tau},\tau)$, as a deterministic function of both the current asset spot, $S_t$,  
and time, $\tau$, consistent with market prices for all options on a given underlying, yielding an asset price model of the type

\begin{eqnarray}
\begin{gathered}
       dS_{\tau} = \mu_{\tau} S_{\tau}\,d\tau + \sigma(S_{\tau},\tau) S_{\tau}\,dW_{\tau}\\
        S_{\tau=0}=S_0.
    \end{gathered}
    \label{eq:local_vol_sde}
\end{eqnarray}
The concept of a LV fully consistent with option markets was developed by Bruno Dupire \cite{dupire1994pricing}, who noted that there is a unique diffusion process consistent with the risk neutral densities, derived from the market prices of European options, see Fig.~\ref{fig:figure_local_vol}. In this sense, this model yields a volatility surface which represents the relationship between an option's implied volatility, its strike price and the time to expiration \cite{local_vol}. By fitting a parametric or non-parametric surface to market data, these models can provide more accurate pricing for options across different strikes and maturities. Thus, this model introduces a dependency of the volatility on the stock prices and the time to maturity, which represent a significant increase of difficulty to solve the model.

We remark that the importance of this model lies in its use to calculate exotic option valuations which are consistent with observed prices of vanilla options, as well as being able to compute the sensitivities ``\textit{greeks}'' of vanilla options.

\begin{figure}[t!]
    \includegraphics[width=0.46\textwidth]{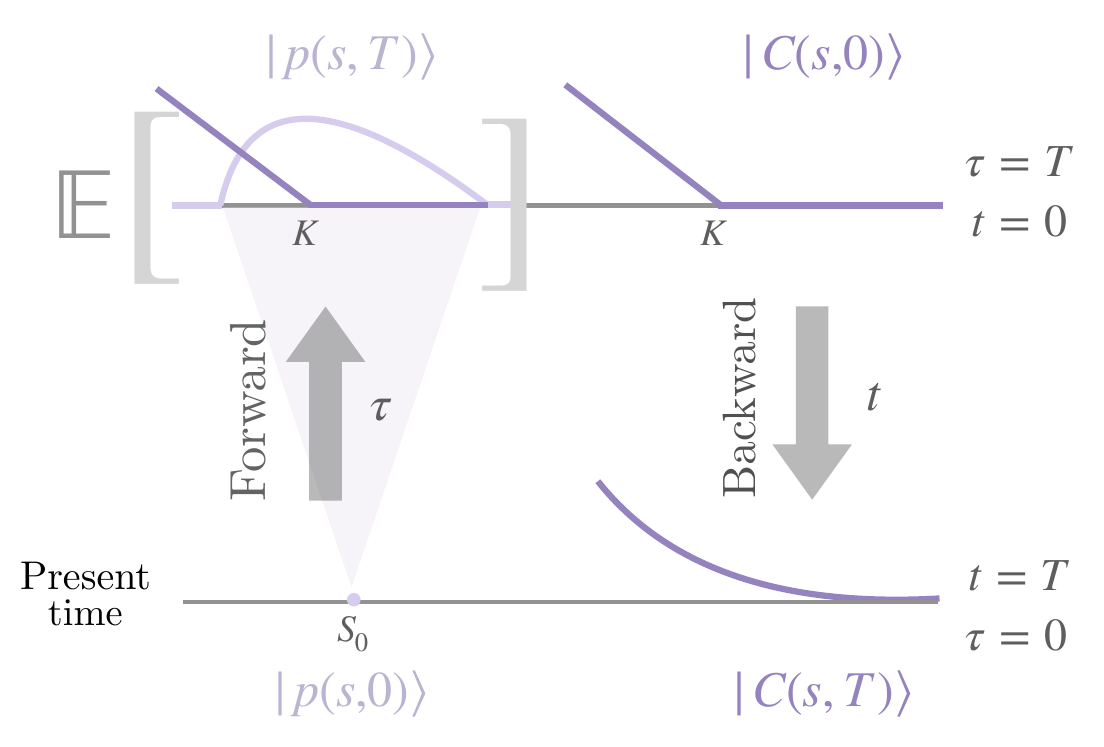}
    \caption{Backward and forward approaches to solve option-pricing, $C(S_0, 0)$. Forward (time $\tau$ from present to future): The stock price is modeled as a stochastic process or as the Kolmogorov forward PDE which evolves from the present time to maturity. In order to retrieve the option price 
    one would need to compute the expected value of the payoff under the resulting underlying price distribution at maturity time and discount it to present time; Backward time $t$ from future to present: The option price is evolved backward from maturity time, given by the final payoff, to present time, through a backward Kolmogorov PDE. }
    \label{fig:figure_3}
\end{figure}

\subsection{Forward and Backward duality in option-pricing}
\begin{figure*}
    \includegraphics[width=.7\textwidth]{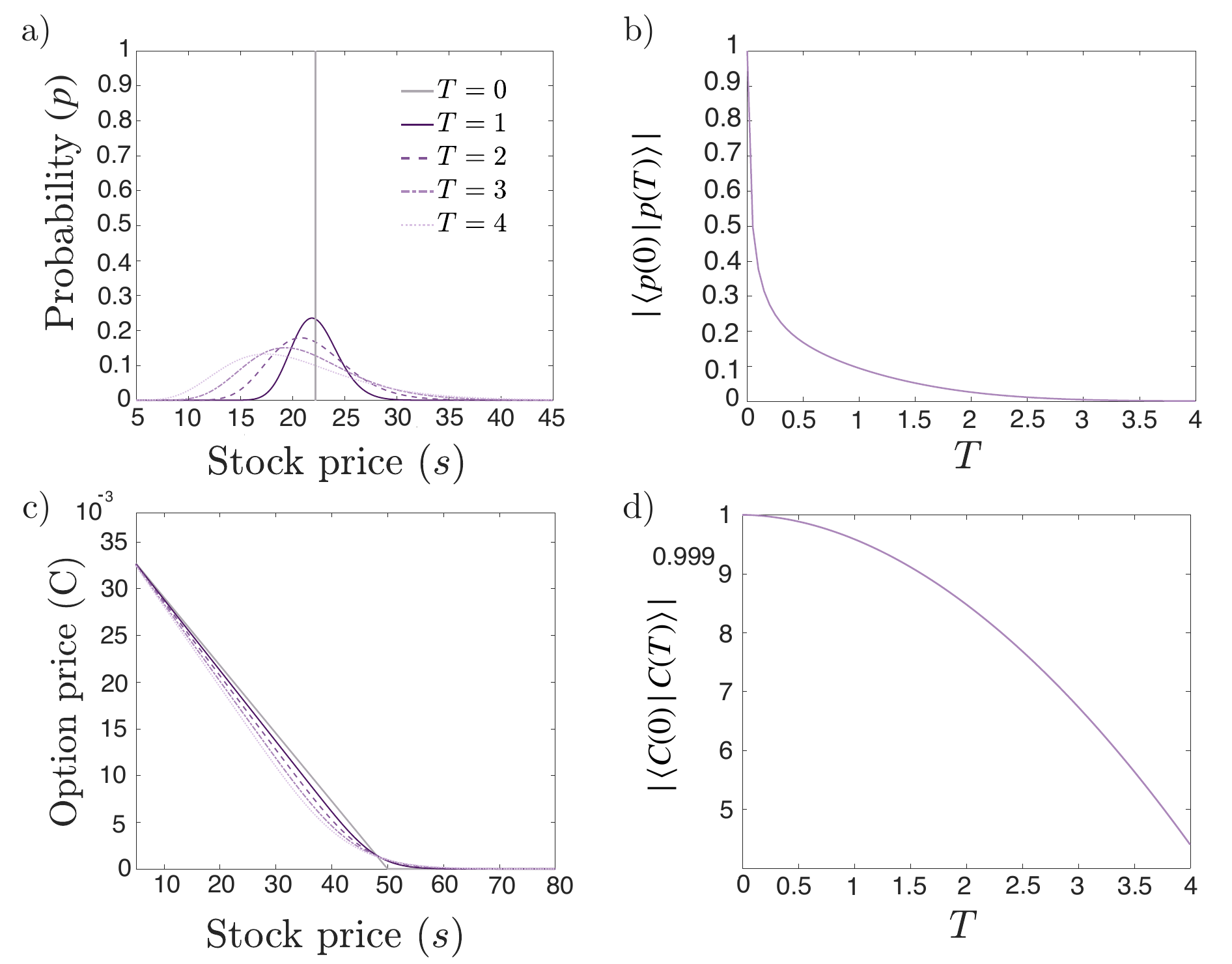}
    \caption{Solution of the Kolmogorov forward and backward  equations from the Black-Scholes model as quantum states: a) Solutions to the forward model are lognormal distributions given by Eq.  (\ref{analytical_forward_sigma_constant}); b) The overlap between the state encoding the solution to the forward equation and the state encoding the initial condition decreases exponentially in time; c) Solutions to the backward model; d) The overlap between the state encoding the solution to the backward equation and the state encoding the initial condition presents a negligible decay in time, remaining at values  $\gtrsim 0.995$ for relevant maturity times. }
    \label{fig:comparative_fw_bw}
\end{figure*}

There exists two different modeling strategies to address the LV model based on either the forward or the backward formulation, see Fig.~\ref{fig:figure_3}. In the forward case, the dynamics of the price of the underlying asset is modeled by either a forward SDE or its equivalent Kolmogorov forward PDE. Following this formulation, the price of the option contract at time $T$ can be calculated as the expected value of the terminal pay-off of the option contract under the resulting underlying asset price distribution, $\mathbb{P}$. On the other hand, in the case of the backward formulation, we straightforwardly obtain the backward Kolmogorov PDE for the option price via the renowned Feynman-Kac formula \cite{black1973pricing}. Note that in the following we are neglecting the discount factor $\exp(-\int_0^T r_d(s) ds)$ needed to achieve the current fair value at present time.

\subsubsection{Backward PDE}
Starting from the SDE model governing the underlying variable, Eq.~(\ref{eq:local_vol_sde}), we can formulate the pricing problem of the option contract as the solution of a backward PDE with an initial condition given by the pay-off function. To that end, we consider the undiscounted price of an European-style derivative contract at $t$, denoted as $C(s, t)$, maturing at time level $\tau = T$ with $\tau := T - t$, i.e. at $t=0$. Here, the variable $t$ ``goes back'' in time from the terminal time $\tau = T$ when $t=0$, while the forward-in-time stochastic processes are now indexed by $S_{\tau}$, starting from $\tau=0$). It can be shown that under the LV model $C(s, t)$ satisfies the \textit{parabolic backward Kolmogorov} PDE
\begin{equation}
\begin{aligned}
\frac{\partial C}{\partial t}
=
  \frac{1}{2} \sigma^2(s, T-t) s^2\frac{\partial^2 C}{\partial s^2}  + \mu_t  s\frac{\partial C}{\partial s}\\
\end{aligned}
\label{eq:SLVPDE_price_undisc}
\end{equation} 
in the open domain $s \in {\rm I\!R}^{+,0}$ and $0 < t \leq T$. The terminal (initial) condition, defined at time level $\tau=T$, is given by $C(s, 0) :=~ C_0(s)$
\begin{figure*}
    \includegraphics[width=1\textwidth]{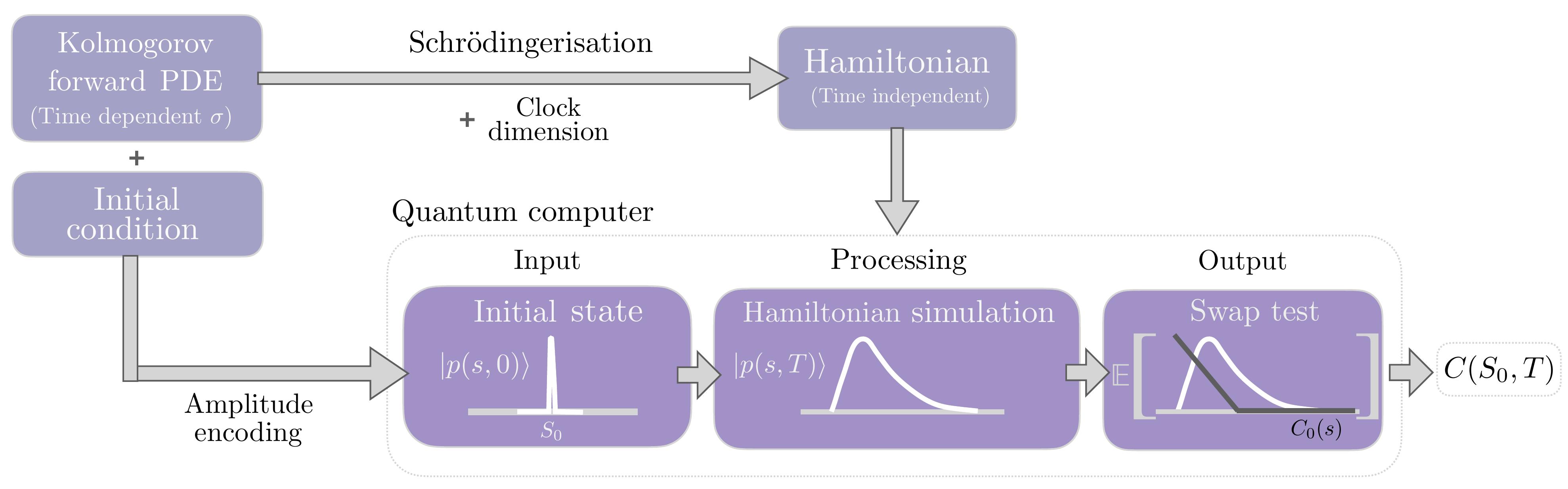}
    \caption{The IPO model is a widely used approach in quantum computing for describing the structure of an end-to-end algorithm. Firstly, classical input data is embedded into a quantum state. Subsequently, the quantum processing involves applying quantum gates and operations to manipulate qubits. In this way properties as superposition and entanglement are key quantum phenomena harnessed during processing, enabling quantum computers to explore many possible solutions to a problem in parallel and achieve a quantum advantage against classical processing, such as in Hamiltonian simulation techniques. Finally, the quantum system is measured and we obtain processed data from the output state, task that is limited by Holevo's bound \cite{holevo}. Every measurement to extract new piece of information will require the execution of the whole IPO model.}
    \label{fig:scheme}
\end{figure*}
\subsubsection{Forward PDE}

The solution  to Eq.~(\ref{eq:SLVPDE_price_undisc}) satisfies 
\begin{equation}
C(s, t) = E^{\mathbb{P}}_{\tau} \left[ C_0(S_{\tau=T}) | S_{\tau=T-t} = s,  \right], 
\label{eq:SLVMC_price_undisc}
\end{equation} 
for $0 \leq t \leq T$. Applying the \textit{tower property of conditional expectations} it follows that
\begin{equation}
\begin{aligned}
C(S_0, T) = E_0 \left[ C(S_{\tau}, t) | S_0 \right] 
 = \int_{-\infty}^{\infty} C(s,  t) p(s, \tau; S_0) ds , 
\end{aligned}
\label{eq:SLVMC_price_undisc_spot}
\end{equation}
where $p(s, \tau=T-t; S_0)$ denotes the density at time $\tau$, of the underlying distribution $S_{\tau}$, given the initial state at spot $S_0$.  It can be shown that the joint density $ p(s, \tau; S_0)$ satisfies the \textit{parabolic Kolmogorov forward} PDE \cite{Conze2008TheFK},

\begin{equation} 
\frac{\partial p}{\partial \tau}=-\frac{\partial}{\partial s}(rsp)+\frac{1}{2}\frac{\partial^2}{\partial s^2}\left( \sigma^2(s, \tau) s^2p\right).
\label{eq:forward_equation_general_form}
\end{equation}
with initial condition $p(s, 0) = \delta(s - S_0)$, and $\delta$ the Dirac delta function. For now on we will assume $r$ is time-independent. The elliptic operator in the forward PDE is simply the dual operator of the backward PDE, from which the fundamental relationship between them is established. From here on and for the sake of simplifying the notion, we denote the time variable $t$ for both the forward and backward problems where no confusion is present in each context.

The above result is general and can now be easily applied in the Black-Scholes model, where volatility is assumed to be constant. In this case, there exists an analytical expression for this density given by a lognormal distribution of the form
\begin{equation}
    p(s,\tau)=\frac{1}{s\sigma\sqrt{2\pi \tau}}\exp{-\frac{(\ln{s}-\gamma)^2}{2\sigma^2\tau}},
    \label{analytical_forward_sigma_constant}
\end{equation}
where $\gamma=(r-\frac{\sigma^2}{2})\tau+\ln{(S_0)}$.

\subsection{Forward vs Backward Models}

Unlike in the backward models, where the value of the contract is calculated for every possible value of the underlying under an specific payoff, forward models compute the resulting price distribution of the underlying from a single spot price, $S_0$. Thus, once the underlying distribution is known at maturity, it can be efficiently used to compute several option prices with different payoffs at once, requiring an integral (expected value) for each pricing problem. Note that the cost of the integral is generally lower than solving the PDE.  In contrast, the backward PDE solves a particular pricing problem at once. This alternative is therefore more efficient when computing one single pricing problem, as it does not require solving an integral as in the forward case, but solely a PDE. However, if multiple pricing tasks must be calculated at once, a backward PDE must be solved for each case (each one with its specific payoff as the terminal condition), leading to a less efficient solution than in the forward case.

On the other hand, in the context of embedding the dynamics of these PDEs into a quantum computer, a key aspect to take into account is the overlap of the normalized  solution (rescaling is needed due to non-unitary evolution) at maturity time, $|C(s,T)\rangle$ $(\ket{p(s,T)})$  with respect to the initial state, $|C(s,0)\rangle$ $(\ket{p(s,0)})$, given by  $|\langle C(s,0)|C(s,T)\rangle|$ $(|\langle p(s,0)|p(s,T)\rangle|)$. A large overlap indicates that the two states to be compared are so close together that more quantum resources are required to resolve these differences. Therefore, in these cases, the dynamics is not suitable for a quantum computer due to the large overhead of resources needed to capture important details in the dynamics (excluding those effects related to a non-unitary evolution that leads to a loss of probability easily captured with an ancillary qubit/qumode). We analyze this issue for both forward and backward PDEs, in a Black-Scholes framework, see Fig.~\ref{fig:comparative_fw_bw}. From this analysis we can appreciate that the backward PDE solutions have a large overlap with the initial state, with ``quantum fidelity'' exceeding $ 0.99$, which makes it very difficult to differentiate between these states. This means that we can capture the non-unitary evolution with a simple single-qubit rotation applied to an ancilla qubit appended to the initial state. Therefore, after normalization we can reproduce a solution with high overlap with respect to the exact solution, and minor details such as the surrounding of the spot corresponding to the strike price are tough to be captured accurately. 

In contrast, \textit{this issue does not arise when solving for the forward equation}, whose overlap of the normalized solution with the initial condition decays rapidly as the maturity time increases. Therefore, we focus on solving the \textit{Kolmogorov forward} PDE, Eq.~(\ref{eq:forward_equation_general_form}), on a quantum computer by  mapping the PDE to a Hamiltonian simulation problem with the Schr\"odingerisation technique \cite{analog, Jin_2024, jin2022quantum, jin2025schrodingerizationmethodlinearnonunitary}.  In this sense, if given the initial condition $|p(s,0)\rangle=|S_0\rangle$ (since we use amplitude encoding with $p(s,0)=\delta(s-S_0)$ and we discretise over $s$)
the Hamiltonian simulation can be performed efficiently,  the simulation of the dynamics of our equation can be efficiently resolved as well, obtaining a quantum state proportional to the solution, $|p(s,T)\rangle$ whose amplitudes are proportional to the price distribution of the underlying asset. This will enable us to apply information retrieval techniques, such as the SWAP test, to compute the expected value of a given pay-off $C_0(s)$ and extract the option price information $C(S_0, \tau=0)$ at the single spot $S_0$ (without the discounting term).

\begin{figure*}[t]
    \includegraphics[width=.95\textwidth]{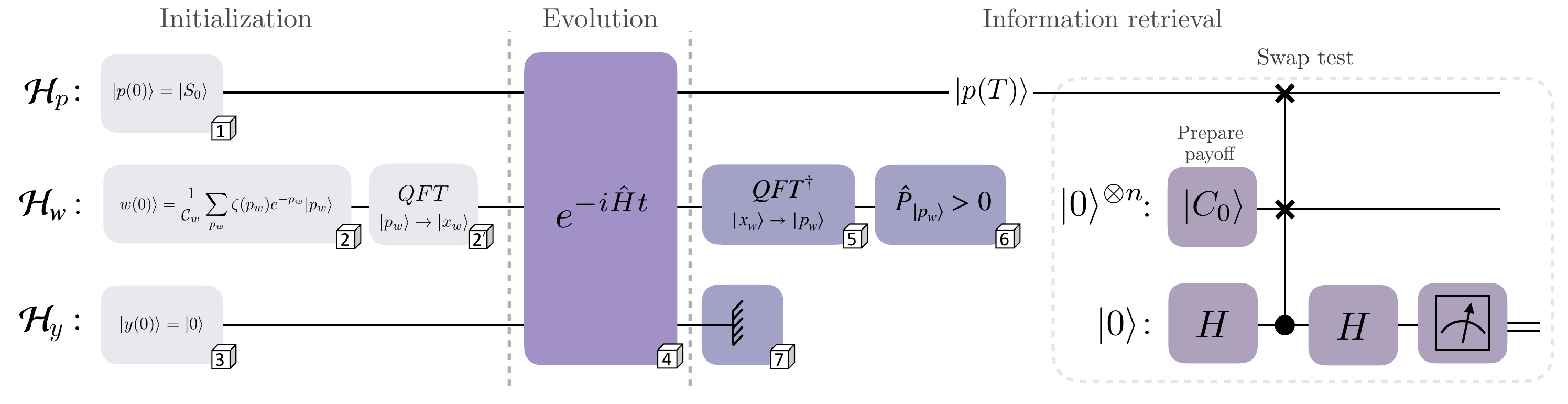}
    \caption{
    General scheme for solving the forward equation with $t$-
and $s$-dependent volatility $\sigma(s, t)$, see Eq.~(\ref{eq: sigma_time_spatial_dependent}). 
(1) The $n$-qubit register is prepared in a state approximating 
a Dirac delta, encoding the initial condition of the forward 
Kolmogorov equation (Eq.~(\ref{eq:forward_equation_general_form})), and is realized in practice 
as a computational basis state corresponding to the current asset price $\ket{S_0}$ or as sharply peaked Gaussian (Eq.~(\ref{eq:gaussian_approx_delta})). 
(2) The auxiliary (Schrödingerisation) qumode is initialized in a $C^\infty$ 
function \eqref{eq:schrodinger_initial_as_multiplication}.  
(2$^\prime$) A quantum Fourier transform is applied to this register 
to switch from momentum to coordinate representation, $\hat{\textbf{p}}_w \rightarrow \hat{\textbf{x}}_w$. 
(3) The auxiliary clock dimension $\tau$ is initialized with a Dirac delta 
function, again approximated by a computational basis state $\ket{0}$ or a sharply peaked Gaussian (Eq.~(\ref{eq:gaussian_approx_delta})). 
(4) The system undergoes unitary evolution under the Hamiltonian in 
Eq.~(\ref{eq: Hamiltonian_forward_equation_s_dependence_onlys}). 
(5) A reversed Fourier transform is applied to switch from coordinate 
to momentum basis, $\hat{\textbf{x}}_w \rightarrow \hat{\textbf{p}}_w$. 
(6) Measurement of positive momentum is performed on the main register.
(7) The ancilla state $\ket{y = t}$ is traced out. 
After steps (6) and (7), the desired state $\ket{p(t)}$ is obtained in the main register. Subsequently, we proceed to encode the payoff in a quantum state, $\ket{C_0}$, and compute a swap test between $\ket{p(t)}$ and $\ket{C_0}$.}
    \label{fig:time_dependent_SCRODINGERISATION}
\end{figure*}

\section{option-pricing on Quantum computers}\label{sec: option-pricing on quantum computers}

In this section, we explore the capability of quantum computers to implement numerical simulations for pricing European Vanilla options under non-constant deterministic volatility regimes, see Fig. \ref{fig:scheme}. Our novel contribution lies in presenting a methodology for solving the Kolmogorov forward PDE, Eq.~(\ref{eq:forward_equation_general_form}), which describes the underlying price dynamics, within the framework of the LV model, see Section \ref{secion: local_vol}. We assume that Dupire's volatility of the underlying asset price, $\sigma$, exhibits a polynomial relationship with both the underlying stock price, $s$, and the time, $t$, and therefore can be described by 

\begin{equation}
\sigma(s,\tau)=\sum\limits_{\substack{k=0,\ \dots,\ D_s\\ q=0,\ \dots,\ D_t}}\mathcal{C}_{k q}s^k \tau^q,\qquad \mathcal{C}_{k q}\in{\rm I\!R}.
    \label{eq: sigma_time_spatial_dependent}
\end{equation}
where $D_s$ and $D_t$ denote the largest polynomial degrees in space and time respectively. This premise is grounded in the observation that thanks to the Weierstrass approximation theorem, any continuous function on a closed interval can be uniformly approximated as closely as desired by a polynomial function.

\begin{remark}[Low-rank separable model for \(\sigma\)]
Eq.~(\ref{eq: sigma_time_spatial_dependent}) treats the local-volatility as a fully
bivariate function \(\sigma(s,t)\). In our setting this exact form is not
implementable: there is no multivariate QSVT/PET that applies a joint
polynomial transform in both \((s,t)\) \cite{Rossi2022multivariable}. Hence a generic bivariate shape cannot be realized as a single multi-input spectral transform with current tools.
To proceed while preserving the intent of Eq.~(\ref{eq: sigma_time_spatial_dependent}),
we adopt a small-rank separable
\[
  \sigma_{\mathrm{lr}}(s,\tau)=\sum_{m=1}^{R}\alpha_m\,r_m(s)\,q_m(\tau),
\]
where \(r_m,q_m\) are 1D basis functions (e.g., piecewise polynomials or Chebyshev
polynomials on a scaled domain). We assume \(R\) is small and independent of
discretization sizes and dimension (e.g., \(R=O(1)\) or at most
\(R=O(\mathrm{polylog}(1/\varepsilon))\)). Each factor is implemented with the
same 1D oracles and combined via a standard LCU, avoiding any need for
multivariate QSVT. Under this assumption, any cost term that queries the
\(\sigma\)-oracle gains a multiplicative factor \(R\) (up to polylogarithmic
precision overheads). For readability, we continue to write \(\sigma(s,t)\), with the understanding
that the implementation uses the small-rank separable representation above.
\end{remark}

\subsection{Set-up}

In order to embed the pricing problem into a quantum computer, we represent the option price of the underlying asset $s$ as the spectrum of the position operator. 
Due to the particular shape of the price probability distribution, see Fig. \ref{fig:comparative_fw_bw} a), choosing a sufficiently wide computational domain ensures $p(x,t)\approx0$ near both endpoints; the boundary dynamics are therefore negligible, and we impose periodic boundary conditions without introducing artifacts; this yields to a simple anti-hermitian form of the derivative operator $\partial/\partial x$ later on. By doing so, we map the value $s$ to the eigenvalue $x$ corresponding to the eigenstate $\ket{x}$ of the position operator, $\hat{\mathbf{x}}$.  In this set-up, the probability distribution $p(x,t)$ corresponding to the solution of the Kolmogorov forward PDE at time $t$ is proportionally encoded into the amplitude distribution of a quantum state,  
\begin{eqnarray}
\begin{gathered}
        \ket{p(\tau)}=\frac{1}{\mathbf{C}}\sum_{x=0}^{2^{n}-1}p(x,\tau)\ket{x}
    \end{gathered}
    \label{eq: wave_function_for_option_price}
\end{eqnarray}
with $\mathbf{C}$ the normalization factor. Note that this amplitude encoding does not preserve the normalization of the probability distribution, $\sum p(x)=1$, but the normalization of its square instead, $\sum |p^2(x)|=1$. Furthermore, implementing the amplitude encoding rather than the \textit{Qsample encoding} \cite{Maria_supervised}, which encodes the square root of the probability distribution into the quantum amplitudes, enables us to perform numerical schemes on the price distribution to simulate the PDE dynamics according to Eq.~(\ref{eq:forward_equation_general_form}). However, it differs in the techniques needed to retrieve the price information, see Section \ref{sec_MCI}.\\

Following the reasoning above,  we firstly perform the amplitude encoding of the initial state $|p(0)\rangle$, and subsequently,  address the simulation of the  dynamics of Eq.~(\ref{eq:forward_equation_general_form}). Due to the structure of this equation, by defining $\hat{\textbf{p}}=-i\partial_x$, it is possible to rewrite it as the \textit{pseudo}-Hamiltonian system

\begin{equation}
\frac{\partial \ket{p}}{\partial \tau} =-i\hat{H}_{LV}\ket{p}, \label{eq:non_unitary_dyn_LV}
\end{equation}
with

\begin{align}
\hat{H}_{LV}&= -i r\hat{\mathds{1}}
  + i\, \sigma^2(\hat{\textbf{x}}, \tau)
  + 4i\, \sigma(\hat{\textbf{x}}, \tau)\, \sigma_x(\hat{\textbf{x}}, \tau)\, \hat{\textbf{x}}
  + i\, \sigma_x^2(\hat{\textbf{x}}, \tau)\, \hat{\textbf{x}}^2\nonumber\\ &
  + i\, \sigma(\hat{\textbf{x}}, \tau)\, \sigma_{xx}(\hat{\textbf{x}}, \tau)\, \hat{\textbf{x}}^2\nonumber - \frac{i}{2} \sigma^2(\hat{\textbf{x}}, \tau)\, \hat{\textbf{x}}^2\, \hat{\textbf{p}}^2\\& + \left[
      r \hat{\textbf{x}}
      - 2 \sigma^2(\hat{\textbf{x}}, \tau) \hat{\textbf{x}}
      - 2 \sigma(\hat{\textbf{x}}, \tau)\, \sigma_x(\hat{\textbf{x}}, \tau)\, \hat{\textbf{x}}^2
    \right] \hat{\textbf{p}}.\label{eq:LV_Hamiltonian}
\end{align}

\noindent Note that we define $\sigma_x(\hat{\textbf{x}}, \tau)$ as the derivative with respect to $x$ of the explicit form of 
$\sigma(x, \tau)$ in Eq.~(\ref{eq: sigma_time_spatial_dependent}). After performing 
this derivative, we promote the classical variable $x$ to the 
operator $\hat{\mathbf{x}}$, so that the resulting quantity is 
denoted as $\hat{\sigma}_x(\hat{\mathbf{x}}, \tau)$, which should be understood as an operator.

One can easily check that the corresponding operator of Eq.~(\ref{eq:LV_Hamiltonian}) is not Hermitian and therefore induces  non-unitary dynamics. However, quantum computers operate under the linear and unitary principles of quantum mechanics. To address this discrepancy, we employ the \textit{Schr\"odingerisation} technique \cite{analog, Jin_2024, jin2022quantum, jin2025schrodingerizationmethodlinearnonunitary}, which has been proved to have an optimal dependence on matrix queries to simulate non-unitary dynamics \cite{jin2025schrodingerizationmethodlinearnonunitary}. This methodology
facilitates our approach to the problem via embedding the dynamics of the PDE into a Hamiltonian simulation problem. In our current approach, we firstly use Schr\"odingerisation to get an Schrodinger PDE, and then discretize operators $\hat{\mathbf{x}}$ and $\hat{\mathbf{p}}$ to get an ODE system.

\subsection{Embedding into a quantum system}\label{subsec:SCRODINGERISATION}

The Schr\"odingerisation method constructs a Hamiltonian, defining it in such a way that its evolution followed by a simple measurement procedure mimics the dynamics of a specific PDE \cite{Jin_2024, PRA2023,analog,PRS2024,cao2023quantum}. In the following, we demonstrate how the non-unitary dynamics of the LV forward model, given by Eq.~(\ref{eq:non_unitary_dyn_LV}) is embedded using this technique. The embedding is performed into a larger Hilbert space through a warped phase transformation, which introduces an additional $w$-qumode represented with $n_w$ ancillary qubits, in such a way that the total dynamics is unitary. 


Moreover, due to the time-dependence of the market volatility, as shown in Eq.~(\ref{eq: sigma_time_spatial_dependent}), the Schr\"odingerisation method leads to a time-dependent Hamiltonian. In this sense, it is widely recognized that such time-dependent Hamiltonians present challenges in implementing the evolution operator \cite{PhysRevLett.114.090502, low2019hamiltoniansimulationinteractionpicture,cao2023quantum}. To address this issue, we employ the technique outlined in Refs.~\cite{cao2023quantum, cao2024unifying}, which involves enlarging the Hilbert space, introducing a \textit{clock dimension}. This expansion transforms the time-dependence of a given Hamiltonian into a position operator $\hat{x}_y$ of an additional qumode represented with $n_y$ ancilla registers. Consequently, this approach enables us to work with a time-independent Hamiltonian, albeit at the expense of slightly enlarging the Hilbert space. Therefore, from now on, we will assume the Hilbert space to be 
\begin{equation}
    \mathcal{H}= \mathcal{H}_p \otimes \mathcal{H}_w \otimes \mathcal{H}_y
    \label{eq:Hilber_space_consist_of_three_dims}
\end{equation}
with $\mathcal{H}_p$ the subsystem that encodes the option price distribution, $\mathcal{H}_w$ the subsystem needed for the Schr\"odingerisation auxiliary qumode (warped phase transformation), and $\mathcal{H}_y$ the subsystem for the auxiliary qumode to remove time-dependence (clock dimension). The overview of both techniques combined is shown in Fig. \ref{fig:time_dependent_SCRODINGERISATION}.

\subsubsection{Schr\"odingerisation}\label{subsec:Schrodingerisation}

We first address the embedding of  the non-unitary dynamics induced by Eq.~(\ref{eq:LV_Hamiltonian}). By applying the Schr\"odingerisation method, we extend our Hilbert space and embed the \textit{pseudo}-Hamiltonian of Eq.~(\ref{eq:non_unitary_dyn_LV}) into a Schrödinger-type equation with unitary evolution governed by the Hamiltonian 
\begin{equation}
\begin{gathered}
     \hat{H}_{LV_S}(\tau) =  \left(
     \frac{\{\hat{C}(\hat{\textbf{x}},\tau),\hat{\textbf{p}}^2\}}{2}
     + \frac{\{\hat{B}(\hat{\textbf{x}},\tau),\hat{\textbf{p}}\}}{2}
     \right)\otimes \hat{\mathds{1}}_w  \\
     +  \left(
     i\frac{[\hat{C}(\hat{\textbf{x}},\tau),\hat{\textbf{p}}^2]}{2}
     + \hat{A}(\hat{\textbf{x}},\tau)
     + i\frac{[\hat{B}(\hat{\textbf{x}},\tau),\hat{\textbf{p}}]}{2}
     \right)\otimes \hat{\textbf{x}}_w,
     \label{eq: Hamiltonian_forward_equation_s_dependence_onlys}
\end{gathered}
\end{equation}
\\
with 
\begin{align*}
 A(\hat{\textbf{x}} , \tau) =\ &
r\hat{\mathds{1}}
- \sigma^2(\hat{\textbf{x}}, \tau)
- 4\, \sigma(\hat{\textbf{x}}, \tau)\, \sigma_x(\hat{\textbf{x}}, \tau)\, \hat{\textbf{x}} \\
&- \sigma_x^2(\hat{\textbf{x}}, \tau)\, \hat{\textbf{x}}^2
- \sigma(\hat{\textbf{x}}, \tau)\, \sigma_{xx}(\hat{\textbf{x}}, \tau)\, \hat{\textbf{x}}^2\\ B(\hat{\textbf{x}}, \tau) =& 
    r \hat{\textbf{x}}
    - 2 \sigma^2(\hat{\textbf{x}}, \tau)\, \hat{\textbf{x}}
    - 2 \sigma(\hat{\textbf{x}}, \tau)\, \sigma_x(\hat{\textbf{x}}, \tau)\, \hat{\textbf{x}}^2\\
    C(\hat{\textbf{x}}, \tau) =& 
    - \frac{i}{2} \sigma^2(\hat{\textbf{x}}, \tau)\, \hat{\textbf{x}}^2.
\end{align*}

\noindent The quantum state describing the system is now $\ket{v(t)}$, and it must be initialized as
\begin{equation}
    \ket{v(0)} =  \ket{p(0)}\otimes \ket{w(0)}  \in \mathcal{H}_p \otimes \mathcal{H}_w.
    \label{eq: Initial condition for Scrodingerisation}
\end{equation} Firstly, we have the initial state $|p(0)\rangle=|S_0\rangle$
represents perfect localization at the initial value $S_0$. However, in numerical and quantum algorithms, this singular state is typically approximated by a sharply peaked, normalized Gaussian function centered at the value $S_0$ for practical reasons. We clarify and elaborate on this construction in the following sections.

\begin{table*}[t]
\centering
\footnotesize
\renewcommand{\arraystretch}{1.65}

\newcommand{\RowHeight}{5.6em}

\begin{tabular}{|>{\centering\arraybackslash}m{0.10\textwidth}| 
                >{\centering\arraybackslash}m{0.27\textwidth}|   
                >{\centering\arraybackslash}m{0.27\textwidth}|   
                >{\centering\arraybackslash}m{0.27\textwidth}|}  
\hline
\parbox[c][\RowHeight][c]{\linewidth}{\centering \textbf{Register}} &
\parbox[c][\RowHeight][c]{\linewidth}{\centering \shortstack{\textbf{Main}\\\textbf{register}}} &
\parbox[c][\RowHeight][c]{\linewidth}{\centering \shortstack{\textbf{Clock}\\\textbf{register}}} &
\parbox[c][\RowHeight][c]{\linewidth}{\centering \shortstack{\textbf{Schr\"odingerisation}\\\textbf{register}}} \\
\hline\hline

\parbox[c][\RowHeight][c]{\linewidth}{\centering \shortstack{\textbf{Initial state}\\\textbf{preparation}}} &
\parbox[c][\RowHeight][c]{\linewidth}{\centering
Computational basis state $|S_0\rangle$\\[4pt]
$\displaystyle
\mathcal{O}\left(n\right)$} &
\parbox[c][\RowHeight][c]{\linewidth}{\centering
Computational basis state $|0\rangle$\\[4pt]
$\displaystyle
\mathcal{O}\left(1\right)$} &
\parbox[c][\RowHeight][c]{\linewidth}{\centering
Smooth cut-off\\[4pt]
$\displaystyle \mathcal{O}\left(n_w \log n_w\right)$} \\
\hline

\parbox[c][\RowHeight][c]{\linewidth}{\centering \textbf{Evolution}} &
\multicolumn{3}{m{0.81\textwidth}|}{\parbox[c][\RowHeight][c]{0.81\textwidth}{\centering
\mbox{$\displaystyle
\mathcal{O}\left(\left(\lVert H \rVert_{\max} T
+ \frac{\log\!\left(1/\epsilon_\text{evol}\right)}{\log\!\left(
e + (\lVert H \rVert_{\max} T)^{-1}\log\!\left(1/\epsilon_\text{evol}\right)
\right)}\right)
\cdot \left(D_s n \log n + n_w \log n_w + D_t n_y \log n_y\right)\right);\;
\lVert H \rVert_{\max} \sim \sigma_{\max}^{2}\, \max\!\left(a,b\right)^{2} 2^{2n}$}
}} \\
\hline

\parbox[c][\RowHeight][c]{\linewidth}{\centering \textbf{Information retrieval}} &
\parbox[c][\RowHeight][c]{\linewidth}{\centering
Swap test\\[4pt]
$\displaystyle \mathcal{O}\left(1/\epsilon_{V_0}^{2}\right)$} &
\parbox[c][\RowHeight][c]{\linewidth}{\centering
Trace out\\[4pt]
$\displaystyle \mathcal{O}\left(1\right)$} &
\parbox[c][\RowHeight][c]{\linewidth}{\centering
$IQFT + P_{>0}$\\[4pt]
$\displaystyle \mathcal{O}\left(\, \frac{\lVert p\!\left(0\right)\rVert^{2}}{\lVert p\!\left(T\right)\rVert^{2}}
\,\log\!\left(1/\epsilon_\text{Schr}\right)\right)$} \\
\hline
\end{tabular}
\caption{Gate complexities for initial state preparation, time evolution, and information retrieval across all registers in Fig.~\ref{fig:time_dependent_SCRODINGERISATION}. Detailed derivations and implementation aspects appear in Section~\ref{subsec:digital_implementation}. We assign $n$ qubits to the main register (encoding the option price), $n_w$ qubits to the Schr\"odingerisation dimension, and $n_y$ qubits to the clock dimension. In the evolution part the sparsity of the Hamiltonian is assumed to be $\mathfrak{s}=5$ and omitted, see detail in Section~\ref{subsec:Hamiltonian simulation -- query complexity}. Initial-state preparation and evolution costs add, while the information-retrieval cost is multiplicative (i.e., the number of repetitions of preparation and evolution). The $\epsilon$ notation is clarified in Table~\ref{table:errors}.}
\label{tab:init_state_prep_evolution_ir_final}
\end{table*}

On the other hand, the auxiliary register $\ket{w(0)}$ encodes the additional variable required to perform the  warped phase transformation introduced by the Schr\"odingerisation method. Preparing this initial state can be an important step, as it can affects both the accuracy and the computational complexity of the quantum simulation. When using the simplest instance of the ancilla state $\psi(p_w)=\exp(-|p_w|)$ one can achieve first-order accuracy. However, one can easily achieve either optimality or near-optimality in precision by modifying this function. For this we need to prepare the function $\psi(p_w)$, constructed as
\begin{equation}
\psi(p_w) = \zeta(p_w) e^{-p_w}.
\label{eq:schrodinger_initial_as_multiplication}
\end{equation}
For optimal scaling with respect to precision $\epsilon_{Schr}$, the cost scales in precision like $\ln(1/\epsilon_{Schr})$, we can choose $\zeta(p_w)=(\text{erf}(ap)+1)/2$ for a constant $a$ \cite{jin2025schrodingerizationmethodlinearnonunitary}. For near-optimality with cost scaling with precision like $\ln(1/\epsilon_{Schr})^{1/\beta}$ (e.g. $\beta=1/2$ using the mollifier below), this can be achieved by using 
\begin{equation}
\zeta(p_w) = (\eta * \chi_{[a_\xi, b_\xi]})(p_w) = \int_{\mathbb{R}} \eta(p_w - z)\, \chi_{[a_\xi, b_\xi]}(z)\, dz
\end{equation} 
which is a smooth cut-off function, obtained by convolving the indicator 
function (window function)
\begin{equation}
\chi_{[a_\xi, b_\xi]}(p_w) = 
\begin{cases}
1, & a_\xi < p_w < b_\xi, \\
0, & \text{otherwise},
\end{cases}
\end{equation}
with the standard mollifier
\begin{equation}
\eta(x) = 
\begin{cases}
\frac{1}{C} \exp\left( \frac{1}{x^2 - 1} \right), & |x| < 1, \\
0, & |x| \geq 1,
\end{cases}
\end{equation}
where $C$ is a normalization constant. The endpoints $a_\xi$ and $b_\xi$ are chosen so that the support of $\zeta(p_w)$
comfortably contains the region of interest for $p_w$, with a buffer for the mollifier’s width.
Typically, $a_\xi$ and $b_\xi$ are selected slightly beyond the computational boundaries to 
ensure rapid decay and avoid boundary artifacts, see Ref.~\cite{jin2025schrodingerizationmethodlinearnonunitary} for details.

As a result, the initial state $\ket{w(0)}$ encodes a discretization of an infinitely differentiable function with compact support, $\psi(p_w)$, which ensures that its Fourier coefficients decay rapidly,  leading  Schr\"odingerisation to achieve either near-optimal complexity or optimal query complexity with respect to precision $\epsilon_\text{Schr}$ when implementing the Hamiltonian simulation Refs.~\cite{jin2025schrodingerizationmethodlinearnonunitary}. Finally, in order to efficiently prepare $\ket{w(0)}$, we leverage the results presented in Thm. \ref{theorem:piecewise_poly} to approximate the desired function by piecewise polynomials.

After unitary evolution under the Schr\"odingerised time-dependent Hamiltonian $\hat{H}_{LV_S}$ for time $T$, the state describing our systems reads 
\begin{equation}
    \ket{v(T)}=\sum_{p_w>0} e^{-p_w}|p(T)\rangle \otimes|p_w\rangle+ \ket{\xi}.
\end{equation}
The retrieval step to obtain the quantum state price distribution, $|p(T)\rangle = \frac{1}{N_p}\sum_{i=0}^{2^n-1} p(x_i)\,|i\rangle$ with $N_p$ the normalization factor, is performed by post-selecting the positive momentum eigenvectors of the auxiliary qumode, as illustrated in Fig.~\ref{fig:time_dependent_SCRODINGERISATION}.

The probability of successfully retrieving the desired dynamics is given by
\begin{equation}
    P_{\text{succ}}\approx L_+(\psi)\frac{\|p(T)\|^2}{\|p(0)\|^2},\qquad L_+=\frac{\int_0^{+\infty}\abs{\psi(p_w)}^2dp_w}{\int_{-\infty}^{+\infty}\abs{\psi(p_w)}^2dp_w}.
\end{equation}
This form arises from the structure of the postselection process,
the measurement projects the quantum state onto the subspace associated
with positive values of $\hat{\bf p}_w$. The numerator
$\|p(T)\|^2$ corresponds to the squared norm of the evolved solution
within the projected subspace reflecting the squared norm ratio as discussed in Ref.~\cite{analog}.

Finally, the total quantum query complexity required to achieve precision $\epsilon_\text{Schr}$ for simulating Eq.~(\ref{eq:non_unitary_dyn_LV}) via  Schr\"odingerisation and including post-selection is
\begin{equation}
\mathcal{O}\left(\log(1/\epsilon_\text{Schr}) \frac{\|p(0)\|^2}{\|p(T)\|^2}\right).
\end{equation}
Note that the the quadratic factor $\|p(0)\|^2/\|p(T)\|^2$ can be further improved to $\|p(0)\|/\|p(T)\|$ with the application of quantum amplitude amplification \cite{Jin_2024}.

\subsubsection{Clock dimension}

We now present a methodology to map the time-dependent Hamiltonian resulting from applying the Schr\"odingerisation technique, given by Eq.~(\ref{eq: Hamiltonian_forward_equation_s_dependence_onlys}), to a time-independent form by introducing an additional auxiliary qumode~\cite{cao2023quantum, cao2024unifying}. In this methodology, the time variable is embedded as the position operator $\hat{\mathbf{x}}_y$ in an enlarged Hilbert space. The resulting Hamiltonian and initial state can be written as
\begin{align} 
    \hat{H} =  \mathds{1} \otimes \mathds{1}_w \otimes \hat{\mathbf{p}}_y + \hat{H}_{LV_S}(\tau = \hat{\mathbf{x}}_y), \label{eq: time-dependent-Hamiltonian}\\
    \ket{g(0)} =  \ket{v(0)} \otimes \ket{y(0)} \in \mathcal{H}, \label{eq:time-dependent hamiltonian +initials}
\end{align}
with $\ket{v(0)}$ given by Eq. (\ref{eq: Initial condition for Scrodingerisation}) and $\ket{y(0)}$ is a representation of localized wave function at zero (delta function). \\

On a discrete grid, the initial conditions for the clock and price registers are the discrete analogue of a Dirac delta function. In practice, one can use the eigenstates in the computational basis $\ket{S_0},\ket{0}$ which is a natural choice. Another way is to approximate the delta function by using a sharply peaked Gaussian distribution, see Section \ref{subsec: initial state prep}.

\subsection{Digitization and circuit implementation}
\label{subsec:digital_implementation}

In this section we present the methodology for the simulation of the the Hamiltonian given by Eq.~(\ref{eq: time-dependent-Hamiltonian}) on a gate-based quantum computer. As discussed above, the probability distribution of the underlying stock price is encoded into an $n$-qubit register and denoted as $\ket{p(\tau)}$. Additionally, we use two $n_w,n_y$ ancillary registers in order to discretize the auxiliary qumodes arising from the Schr\"odingerisation and clock dimension methodology. However, the temporal dimension remains \textit{analog}, as the time $\tau$ is interpreted by the continuum quantum circuit's parameters. Thus, our methodology comprises three segments: (i) Initial state preparation, (ii) implementation of Hamiltonian simulation, and (iii) information retrieval. The computational cost of every of these steps is depicted in Table~\ref{tab:init_state_prep_evolution_ir_final}. For the convenience of the reader we also provide Table~\ref{table:errors}.

\begin{table}[b!]\label{table:errors}
\footnotesize
\setlength{\tabcolsep}{3pt}
\renewcommand{\arraystretch}{1.05}
\centering
\begin{tabular}{@{}l|| l | l@{}}

\textbf{Symbol} & \textbf{Meaning} & \textbf{First seen} \\
\hline
$\epsilon_{\mathrm{prep}}$ &
\shortstack{Target accuracy for the state preparation} &
Thm.~\ref{theorem:piecewise_poly} \\

$\epsilon_{\mathrm{evol}}$ &
\shortstack{Hamiltonian simulation tolerance} &
Thm.~\ref{theorem: Random walk} \\

$\epsilon_{\mathrm{Schr}}$ &
\shortstack{The error caused by Schrodingerisation technique} &
Sec.~\ref{subsec:Schrodingerisation}\\

$\epsilon_{\mathrm{V_0}}$ &
\shortstack{Information retrieval noise in the option price} &
Sec.~\ref{subsec:computing the payoff} \\

\end{tabular}
\caption{Notation for main $\epsilon$-errors and their first appearance. The are are two kind of errors: $\epsilon_{\mathrm{prep}}$ and $\epsilon_{\mathrm{evol}}$ are caused by truncating series approximating target functions (like Taylor series). The errors $\epsilon_{\mathrm{Schr}}$ and $\epsilon_{\mathrm{V_0}}$ are absolute errors.}
\end{table}

\subsubsection{Initial state preparation}\label{subsec: initial state prep}

The total Hilbert space of our quantum system is given by
Eq.~(\ref{eq:Hilber_space_consist_of_three_dims}) and consists of three registers
of dimensions $2^n$, $2^{n_w}$, and $2^{n_y}$, corresponding to the stock price and the
two auxiliary variables respectively. To efficiently prepare initial quantum states for these
registers, we employ a theorem that enables an efficient construction of
quantum states encoding piecewise polynomial functions on a qubit register.

\begin{theorem}[Efficient preparation of piece-wise polynomial functions~\cite{guseynov2024efficient,gonzalez2024efficient}]
\label{theorem:piecewise_poly}
Let $f(x)$ be a piece-wise continuous function
$f: \mathbb{R} \rightarrow \mathbb{C}$ that can be decomposed
into $G$ pieces, each described by a degree-$Q_g$ polynomial:
\[
f(x) =
\begin{cases}
f_1(x) = \sum_{i=0}^{Q_1} \alpha^{(1)}_i x^i, & \text{if } K_1 \geq x \geq a \\
f_2(x) = \sum_{i=0}^{Q_2} \alpha^{(2)}_i x^i, & \text{if } K_2 \geq x > K_1 \\
\vdots \\
f_G(x) = \sum_{i=0}^{Q_G} \alpha^{(G)}_i x^i, & \text{if } b \geq x > K_{G-1}
\end{cases}
\]
where $\alpha_i^{(g)} \in \mathbb{C}$.

Then, there exists a $n$-qubit quantum circuit $U_f$ that
efficiently prepares a $2^n$-dimensional discretized quantum state
proportional to $f(x)$, using:
\begin{enumerate}
    \item $\mathcal{O}\left(\sum_{g=1}^G Q_g n \log n\right)$ C-NOT and single-qubit gates,
    \item $n-1$ ancilla qubits initialized in the $\ket{0}$ state and returned to $\ket{0}$ by the end of the circuit.
\end{enumerate}
with success probability proportional to the filling ratio
\[
\mathcal{F} := \|f\|_2^2/(2\|f\|^2_{\max}).
\]
\end{theorem}

\vspace{0.5em}
\noindent

We now apply Theorem~\ref{theorem:piecewise_poly} to the three
initial states needed for our quantum algorithm, given by Eq. (\ref{eq:time-dependent hamiltonian +initials}). The initial states for the price distribution $\ket{p(0)}$ and the clock dimension $\ket{y(0)}$, are both modeled as delta functions, which can be naturally implemented by the computational basis states $\ket{S_0},\ket{0}$. In particular, the initial state for the main register $\ket{S_0}$ can be implemented by no more than $n$ $X$ gates; and the state $\ket{0}$ for the clock dimension requires $0$ gates.

An alternative way is to use a Gaussian distribution function
\begin{equation}
    \delta_\omega(x_y) \simeq \left( \frac{1}{2 \pi \omega^2} \right)^{1/4} \exp\left(-\frac{x_y^2}{4 \omega^2}\right).
    \label{eq:gaussian_approx_delta}
\end{equation}
which can be realized as a single-piece polynomial $$q_Q(x)=\frac{1}{\omega\sqrt{2\pi}}\sum_{s=0}^Q\frac{(-1)^s}{s!}\left(\frac{x^2}{2\omega^2}\right)^s$$ of degree
$Q = \mathcal{O}\left(\frac{ \log\left(1/\epsilon_\text{prep}\right) + \log \left(1/\omega\right) }{ \log\left(1+2\omega\log\left(1/\epsilon_\text{prep}\right)\right)}\right)
$,
leading to a circuit complexity
$$\mathcal{O}\left( n \log n \cdot 
\frac{ \log\left(1/\epsilon_\text{prep}\right) + \log \left(1/\omega\right) }
     { \log\left(1 + 2\omega \log\left(1/\epsilon_\text{prep}\right)\right) } 
\right),$$
where $n$ ($n_y$) is the number of qubits in the corresponding register.

On the other hand, for the Schr\"odingerisation dimension, the initial state is the smooth $\psi(p_w)$ which can be efficiently approximated by a low degree polynomial. Therefore, all initial states ($\ket{p(0)}$, $\ket{w(0)}$, and $\ket{y(0)}$) can be efficiently prepared with the quantum gate complexities
specified in Theorem~\ref{theorem:piecewise_poly}.This will employ a number of resources determined by the degree of the polynomial used for the approximation in each subspace.

\subsubsection{Hamiltonian simulation -- query complexity}\label{subsec:Hamiltonian simulation -- query complexity}

Now, we focus on how to implement the unitary $e^{-it\hat{H}}$, where $\hat{H}$ is given by Eq.~(\ref{eq: time-dependent-Hamiltonian}). This step requires introducing discrete versions of the quantum coordinate and momentum operators, $\hat{\mathbf{x}}$ and $\hat{\mathbf{p}}$.\\

For the position operator defined over the interval $(a, b)$ with a uniform grid, the corresponding 
matrix representation in its own basis takes the diagonal form
 \begin{equation}
    \hat{X}=\left(\begin{array}{ccccc}
a & 0  & \dotsm & 0 & 0 \\
0 & a+\Delta_x  &\dotsm& 0 & 0\\
\rotatebox[origin=c]{270}{\dots}&&\rotatebox[origin=c]{-45}{\dots}&&\rotatebox[origin=c]{270}{\dots}\\ 
0 & 0 &\dotsm & b-\Delta_x &0\\
0 & 0 &\dotsm & 0 &b\\
\end{array}
\right).
\label{eq:x_operatorCBQC}
\end{equation}
Here, $\Delta_x=\frac{b-a}{2^n-1}$ is the grid size.

\setlength{\arrayrulewidth}{0.3pt}
\renewcommand{\arraystretch}{1.9}

\begin{table}[H]
\centering
\footnotesize
\begin{tabular}{c||c||c}

\textbf{Der. order} & \textbf{Acc. order} & \textbf{Formula} \\ 
\hline
\hline
First & Second & 
$\frac{\partial u}{\partial x}\big|_{x_i} \approx \frac{u_{i+1} - u_{i-1}}{2\Delta x} + \mathcal{O}(\Delta x^2)$ \\
\hline
Second & Second & 
$\frac{\partial^2 u}{\partial x^2}\big|_{x_i} \approx \frac{u_{i-1} - 2u_i + u_{i+1}}{\Delta x^2} + \mathcal{O}(\Delta x^2)$ \\
\hline
First & Fourth & 
$\frac{\partial u}{\partial x}\big|_{x_i} \approx \frac{-u_{i+2} + 8u_{i+1} - 8u_{i-1} + u_{i-2}}{12\Delta x} + \mathcal{O}(\Delta x^4)$ \\
\hline
Second & Fourth & 
$\frac{\partial^2 u}{\partial x^2}\big|_{x_i} \approx \frac{-u_{i+2} + 16u_{i+1} - 30u_i + 16u_{i-1} - u_{i-2}}{12\Delta x^2} + \mathcal{O}(\Delta x^4)$ \\
\end{tabular}
\caption{
Central finite-difference schemes for the first and second derivatives of $u(x)$ at $x_i$, including second- and fourth-order accuracy.
}
\label{table:central-difference-schemes}
\end{table}

Regarding the discrete representation of the momentum operator,  the second order central difference scheme with periodic boundary conditions is the most suitable choice for the purposes of our quantum simulation (see Table~\ref{table:central-difference-schemes}) \cite{guseynov2024Hamsim}. It yields a symmetric, Hermitian matrix with only two nonzero entries per row, which is crucial for the efficiency of the Hamiltonian simulation. This reads

\begin{equation}
\hat{P}= -\frac{i}{2\Delta x}
\begin{pmatrix}
0 & 1 & 0 & \cdots & -1 \\
-1 & 0 & 1 & \cdots & 0 \\
0 & -1 & 0 & \cdots & 0 \\
\vdots & \vdots & \vdots & \ddots & 1 \\
1 & 0 & 0 & \cdots & 0 \\
\end{pmatrix}.
\label{eq:momentum1}
\end{equation}

\noindent The discrete representation of the second-derivative operator reads
\begin{equation}
\hat{P}^{2}=-\frac{1}{\Delta x^2}
\begin{pmatrix}
-2 & 1 & 0 & \cdots & 1 \\
1 & -2 & 1 & \cdots & 0 \\
0 & 1 & -2 & \cdots & 0 \\
\vdots & \vdots & \vdots & \ddots & 1 \\
1 & 0 & 0 & 1 & -2 \\
\end{pmatrix}.
\label{eq:momentum2}
\end{equation}

As we have already discussed, due to the particular shape of the price probability distribution $p(x,t)$, see Fig. \ref{fig:comparative_fw_bw} a), if we consider a sufficiently large domain, both boundaries will be zero and therefore we can take periodic boundary conditions when defining the derivative operator,
which guarantee that the momentum operator, $\hat{P}$, remains Hermitian. Nevertheless, as shown in Ref.~\cite{guseynov2025quantum}, other boundary conditions such as Robin, Dirichlet, or Neumann are
also implementable and do not change the main scaling in terms of
$n$ and $\epsilon$.

The respective discrete version of the operators $\hat{\mathbf{x}}_w,\hat{\mathbf{p}}_w, \hat{\mathbf{x}}_y,\hat{\mathbf{p}}_y$ is defined accordingly with the corresponding number of qubits in every subspace.
Now, we have all the ingredients for implementing the Hamiltonian simulation.
The following theorem provides the foundation for efficient simulation
of sparse Hamiltonians on a qubit-based quantum computer.

\begin{theorem}[Optimal sparse Hamiltonian simulation using queries (Theorem 3 from \cite{low2017optimal})]
    \label{theorem: Random walk}
    A $\mathfrak{s}$-sparse Hamiltonian $H$ acting on $n+n_w+n_y$ qubits with matrix elements
    specified up to $m$-bit precision can be simulated for time $t$ within error $\epsilon_\text{evol}$, and success probability at least $1-2\epsilon_\text{evol}$ with
    \[
    \mathcal{O}\left(\gamma+\frac{\log(\gamma/\epsilon_\text{evol})}{\log\log(\gamma/\epsilon_\text{evol})}\right)
    \]
    queries and a factor
    \[
    \mathcal{O}\left(n+n_w+n_y+m\,polylog(m)\right)
    \]
    additional quantum gates, where $\gamma := \mathfrak{s} \lVert H \rVert_{\max} t.$
    The oracles have the form
    \[
    O_H\ket{j,k,z} = \ket{j,k,z\oplus H_{jk}};\quad O_F\ket{j,l} = \ket{j, f(j,l)},
    \]
    where $j$ is a row number, $k$ a column number, $H_{jk}$ the corresponding matrix element,
    and $f(j,l)$ a function giving the column index of the $l$-th nonzero element in row $j$.
\end{theorem}

The next step is to analyze the parameter $\gamma = {\mathfrak{s}} \lVert H \rVert_{\max} t$
from Theorem~\ref{theorem: Random walk}. The sparsity, $\mathfrak{s}$, is determined by the order of the finite-difference stencils and the degree of the momentum operator contributions, as described in Eqs.~(\ref{eq:momentum1}) and (\ref{eq:momentum2}). In particular, with our choice, the Hamiltonian for the option-pricing problem is such that each row in the matrix representation of $\hat{H}$ contains at most three nonzero elements. Therefore, using the second order accurate central finite-difference scheme from Table~\ref{table:central-difference-schemes} we get $\mathfrak{s} = 5$ ($3$ for the main register and additional $2$ from the clock dimension).  

The largest matrix elements arise from the term
$\hat{\sigma}^2(t)\hat{X}^2\hat{P}^2$
(see Eq.~\ref{eq: Hamiltonian_forward_equation_s_dependence_onlys}).
Here, $\sigma(x,t)$ is the local-volatility function defined in Eq.~(\ref{eq: sigma_time_spatial_dependent}),
and we denote
\[
\sigma_{\max} := \sup_{x,\,t}\, \sigma(x,t)
\]
as its maximum over both spatial and temporal parameters. The maximal value of the coordinate operator is determined by the grid endpoints,
so $\max(\hat{\mathbf{x}}^2) \sim \max(\abs{a},\abs{b})^2$.
In a naive approximation without considering any kind of ultraviolet cutoff for the momentum operator, the largest entry in the discrete representation is
$\lVert \hat{\mathbf{p}}^2 \rVert_{\max} \sim 1/\Delta_x^2 \sim 2^{2n}$,
where $n$ is the number of qubits per register.

Combining these and considering the number of ubits per qumode is the same, the overall maximum norm is
\[
\lVert H \rVert_{\max} \sim \sigma_{\max}^2\, \max(\abs{a},\abs{b})^2\, 2^{2n}.
\]
Thus, the query complexity parameter in the theorem is
\[
\gamma \sim \mathfrak{s}\, \sigma_{\max}^2\, \max(\abs{a},\abs{b})^2\, 2^{2n}\, T.
\]

\begin{remark}\label{remark:p scaling}
The scaling $\lVert \hat{P}^2 \rVert_{\max} \sim 2^{2n}$, arising from the discretized
momentum operator, is the most restrictive factor limiting quantum advantage.
As a result, the achievable speedup in the number of spatial grid points $N=2^n$
is at most polynomial, because the norm of the Hamiltonian—dominated by
the momentum term—appears multiplicatively in the total query and gate complexities. In some set-ups this problem could be circumvented by proposing energy cutoffs that lead to accurate polynomial scaling effective descriptions \cite{Rolando_low} or using a second quantization for describing the position and momentum operators \cite{bravyi2025quantumsimulationnoisyclassical}.
\end{remark}

\subsubsection{Hamiltonian simulation -- gate complexity}

Many optimal quantum simulation algorithms, such as quantum walks with signal processing~\cite{low2017optimal} or those based on fractional queries~\cite{berry2014exponential}, assume access to an oracle that outputs Hamiltonian matrix elements in binary form. However, the explicit construction of these oracles using quantum circuits can be resource-intensive, especially for Hamiltonians with complex coefficients ~\cite{munoz2018t,haner2018optimizing}.

In this regard, Theorem~\ref{theorem: Random walk} provides complexity bounds in terms of quantum queries, which are not directly comparable to circuits described in terms of classical logic gates such as NOT, AND, or OR. In this sense, the construction of the required oracles, $O_H$ and $O_F$, presumably efficient, is left unspecified, which makes a practical implementation and a fair comparison with classical algorithms cumbersome.

To overcome this limitation, we suggest employing a strategy based on block-encoding, as in Refs.~\cite{guseynov2024Hamsim,guseynov2025quantum}.
This method enables efficient Hamiltonian simulation for PDEs with piece-wise continuous coefficients, without the need for
complicated digital oracles.

\begin{theorem}[Quantum simulation of the option-pricing Hamiltonian (Theorems 8 \& 9 from \cite{guseynov2025quantum})]
\label{theorem:QuantumHamiltonianSimulation}
Let $H$ be a $\mathfrak{s}$-sparse Hamiltonian as in Eq.~(\ref{eq: time-dependent-Hamiltonian}),
acting on $n + n_y + n_w$ qubits, where $n$ encodes the discretized
asset price, $n_y$ encodes the time (clock) dimension, and $n_w$ is the auxiliary
register for Schr\"odingerisation. Assume the volatility $\sigma(x, t)$ is defined in Eq.~(\ref{eq: sigma_time_spatial_dependent}).
Then, for any precision $\epsilon_\text{evol} > 0$ and evolution time $T > 0$, we can implement an $\epsilon_\text{evol}$-precise block-encoding of the evolution operator $e^{-iHT}$, with
block-encoding parameters
\[
(2,\, c,\, \epsilon_\text{evol}), \quad
c = \mathcal{O}\left( \log n + \log n_y + \log n_w \right).
\]

\noindent The total gate complexity (C-NOTs and one-qubit rotations) scales as
\[
\mathcal{O}\left(\underbrace{\left[
\mathfrak{s}\, \lVert H \rVert_{\max}\, T
+ \frac{\log(1/\epsilon_\text{evol})}{\log\left( e + \frac{\log(1/\epsilon_\text{evol})}{\mathfrak{s}\,\lVert H \rVert_{\max}\, T} \right)}
\right]}_{\text{Optimal Ham. sim.}}\right.
\]
\[
\times \underbrace{\left[
D_s n \log n + \mathfrak{s} n + n_w \log n_w + D_t n_y \log n_y
\right]}_{\text{Block encoding of } H}\Biggl)
\]
where $\mathfrak{s} = 3$ for the central-difference discretization and $\lVert H \rVert_{\max} \sim \sigma_{\max}^2 \max(a, b)^2 2^{2n}$,
with $\sigma_{\max} := \max_{x, t} \sigma(x, t)$. This is valid for $\mathfrak{s}\, \lVert H \rVert_{\max}\, T=\mathcal{O}\Bigl(\frac{\log(1/\epsilon_\text{evol})}{\log\left( e + (\mathfrak{s}\,\lVert H \rVert_{\max}\, T)^{-1}\log(1/\epsilon_\text{evol}) \right)}\Bigl)$.
\end{theorem}

This result shows that quantum algorithms for option-pricing
PDEs achieve a genuine complexity advantage, with quantum gate complexity scaling as $\mathcal{O}(N^2 \log N \log \log N)$ versus the classical $\mathcal{O}(N^3)$ cost (see Appendix~\ref{appendix: comparison with classical}), with $N=2^n$ the number of grid points. These are both simplified scaling expressions.
The quantum advantage does not depend on black-box oracles, but arises from the explicit construction of block-encodings. These results illustrate a polynomial speedup and encourage further study of resource-efficient quantum circuits for financial PDEs. Additionally, a refinement to study effective regimes of energy could lead to a better speedup, see Remark~\ref{remark:p scaling}.

The multidimensional cost scaling of the option–pricing problem under our framework remains polynomial in the dimension. In terms of the Hamiltonian–simulation methodology, the most restrictive factor is the operator norm $\lVert H\rVert_{\max}$, which is largely controlled by the maximal order $\varkappa$ of the momentum (derivative) operator in the Kolmogorov/Black–Scholes setting \eqref{eq:forward_equation_general_form}; here $\varkappa=2$. This persists in multi–asset models \cite{kohn2011-section1}. For instance, the two–asset forward Kolmogorov PDE \cite{Antoine_Conze} for $p=p\!\left(S_1,S_2,t\right)$ takes the form
\begin{equation}
\begin{aligned}
\partial_t p
&+ \sum_{i=1}^{2}\left(
\partial_{S_i}\left(\mu_i S_i p\right)
- \frac{1}{2}\partial_{S_iS_i}\left((\sigma_i S_i)^2 p\right)
\right) \\
&- \rho \partial_{S_1S_2}\left(\sigma_1 S_1 \sigma_2 S_2 p\right)
= 0.
\end{aligned}
\label{eq:twoD_forward_FP_C}
\end{equation}
We discretize the state space with a total of $d n$ qubits using $n$ qubits per asset, corresponding to a per–asset grid of $2^n$ points and an overall Hilbert–space dimension $2^{d n}$.

\begin{theorem}[Multidimensional scaling and simulation cost]\label{thm:multidim_scaling}
Consider a $d$–asset local–volatility model in the forward (Kolmogorov) formulation, discretized with $n$ qubits per asset (total $d n$ qubits). Then the gate complexity of one time–evolution step via Hamiltonian simulation scales polynomially in $d$ and, in particular, as $\mathcal{O}\!\left(d^{5}\right)$ in the $d$–dependence. 
\end{theorem}
\begin{proof}
This arises from:
\begin{enumerate}
    \item \textbf{Operator norm.} The maximum derivative order remains $\varkappa=2$ independent of $d$, while the number of drift and diffusion terms grows like $d^{2}$ (including $d$ diagonal and $\binom{d}{2}$ cross–derivative terms), so $\lVert H\rVert_{\max}$ scales as $\Theta\!\left(d^{2}\right)$.
    \item \textbf{Sparsity.} If a one–dimensional stencil has sparsity $\mathfrak{s}$, then tensor–product and cross terms yield an overall sparsity $\Theta\!\left(d^{2}\,\mathfrak{s}^{2}\right)$.
    \item \textbf{Block encoding price.} The cost of a single block–encoding query is linear in $d$ (see Theorem~\ref{theorem:QuantumHamiltonianSimulation} and \cite{guseynov2024Hamsim}).
\end{enumerate}
Combining these contributions gives the stated $\mathcal{O}\!\left(d^{5}\right)$ evolution cost in the dimension $d$.
Moreover, the complexity of preparing the initial condition scales linearly in $d$:
for $p\!\left(S_1,\dots,S_d,0\right)=\prod_{j=1}^{d}\delta\!\left(S_j-S_{j,0}\right)$ one has $\ket{p(0)}=\bigotimes_{j=1}^{d}\ket{S_{j,0}}$, so preparing each $\ket{S_{j,0}}$ independently on its $n$–qubit register yields overall $\mathcal{O}\!\left(dn\right)$ cost (cf. Section~\ref{subsec: initial state prep}).
\end{proof}

\subsubsection{Computing the payoff}\label{subsec:computing the payoff}
\label{sec_MCI} 
The final task for solving the forward equation, Eq. (\ref{eq:forward_equation_general_form}), is computing the option price, Eq. (\ref{eq:payoff_definition}). In the forward model, the quantity of interest is encoded as an expected value, while in the backward formulation, it is codified in an specific amplitude of the wave function. This distinction is crucial when the discretization is refined: quantities encoded as expected values converge to their \textit{continuous limit}, while those encoded in amplitudes decay exponentially towards zero, resulting into a prohibitive number of samples to estimate its value. 

Here, we address the most common payoff corresponding to put and call options, see Table \ref{tab:payoff}. For the sake of simplicity, hereafter we consider only the put option payoff. Thus, the option price reads
\begin{equation}
    V(0)=e^{-rT}\sum_{i=0}^{\kappa}p(x_i,T; S_0)(K-x_i)\Delta x,
    \label{eq:CBQCpayoff_integral}
\end{equation}
where $\kappa =\text{max}_j \{x_j \leq K\}$. We define the state encoding the discretized normalized payoff as 
\begin{equation}
    \ket{C_0}=\frac{1}{N_f}\sum_{i=0}^{\kappa}(K-x_i)\ket{i},
    \label{eq:auxiliary wave function swap test}
\end{equation}
with $N_f$ the normalization factor. We show how to construct $\ket{C_0}$ with probability $\mathcal{O}(1)$ in  Appendix \ref{appendix: auxiliary wave function}. The total complexity of $\ket{C_0}$ construction is $6n^2-4n+6$ CNOTs and $n-1$ ancillas.  \\

Since the quantity of interest is an expected value between the price distribution and the payoff, the option price can be calculated efficiently using the swap test. Note that in our case, we encode the probability distribution into the amplitudes of a quantum state, differing from previous approaches where the amplitudes encode the square root of the price \cite{woerner_option_pricing, Woerner_threshold, Neufeld_CPWA, Montanaro_multilevel_MC_local_vol} and therefore can benefit from the quadratic speedup of quantum Monte-Carlo integration
\cite{Montanaro_MC_original}.

\begin{theorem}[Swap test (Proposition 6, \cite{huang2019near})]\label{theorem:swap test}
    Given multiple copies of $n$-qubit quantum states $\ket{p(T)}$ and $\ket{C_0}$, there is a quantum algorithm that determines the overlap $|\braket{C_0}{p(T)}|^2=\left(\sum_{i=0}^{\kappa}p(x_i,T; S_0)(K-x_i)\right)^2/(N_f N_p)^2$ to additive accuracy $\epsilon_{st}$ with failure probability at most $\delta$ using $\mathcal{O}(\frac{1}{\epsilon_{st}^2}\log(\frac{1}{\delta}))$ copies and the quantum circuit with $7n$ CNOT operations.\\
\end{theorem}

\noindent
Additionaly, we also require and additional swap test to determine $|\braket{p(T)}{+}^{\otimes n}|^2$. By using the derivations in Appendix~\ref{sec:appendix_swap}, the  
final closed-form option prices for call and put payoffs are summarized in Table \ref{tab:TableIV}.

\begin{remark}[Swap-test retrieval for multi-asset pricing]
In the swap-test setting (Theorem~\ref{theorem:swap test}), the
risk-neutral price of a $d$-asset European claim can be written as
\begin{equation}
V(0)=e^{-rT}\!\int_{\mathbb{R}_+^d} C_{d,0}(\mathbf s)\,p(\mathbf s,T)\,d\mathbf s,
\qquad \mathbf s=(s_1,\ldots,s_d),
\label{eq:multi_asset_integral_price_C}
\end{equation}
where $C_{d,0}(\mathbf s)$ is the multi-asset payoff. After discretization,
one prepares two amplitude-encoded states
\begin{equation}
\ket{C_{d,0}}\propto\sum_{\mathbf j} C_{d,0}(\mathbf s_{\mathbf j})\,\ket{\mathbf j},
\qquad
\ket{p_d}\propto\sum_{\mathbf j} p(\mathbf s_{\mathbf j},T)\,\ket{\mathbf j},
\end{equation}
and uses the swap test to estimate $|\langle C_{d,0}|p_d\rangle|^2$,
which yields the discretized version of
\eqref{eq:multi_asset_integral_price_C} up to known scaling factors.
The main cost here lies in state preparation; the swap test itself is
comparatively lightweight.

Representative multi-asset payoffs include \cite{hull2016options,Johnson1987MaxMin}:
\begin{itemize}
\item \textbf{Arithmetic basket (call/put), weights $w_i$, strike $K$:}\\
$C_{d,0}(\mathbf s)=\big(\sum_{i=1}^d w_i s_i-K\big)^{+}$,\quad
$C_{d,0}(\mathbf s)=\big(K-\sum_{i=1}^d w_i s_i\big)^{+}$.
\item \textbf{Spread / exchange (Margrabe-type):}\\ \quad
$C_{d,0}(\mathbf s)=\big(s_1-\kappa s_2-K\big)^{+}$.
\item \textbf{Geometric basket:}\\ \quad
$C_{d,0}(\mathbf s)=\big(\prod_{i=1}^d s_i^{\,w_i}-K\big)^{+}$.
\item \textbf{Digital (cash-or-nothing):}\\ \quad
$C_{d,0}(\mathbf s)=\mathbf 1\{\sum_{i=1}^d w_i s_i\ge K\}$,\quad
$C_{d,0}(\mathbf s)=\mathbf 1\{\max_i s_i\ge K\}$.
\end{itemize}

We note that from the view of the representative payoffs the corresponding quantum states are either factorized or well-approximated by a short
linear combination of factorized terms:
\begin{equation}
\ket{C_{d,0}}=\sum_{i=0}^{\mathcal Q}\alpha_i\,
\ket{C^{(i)}_{1,0}}\otimes\ket{C^{(i)}_{2,0}}\otimes\cdots\otimes\ket{C^{(i)}_{d,0}},
\qquad \mathcal Q\ \text{small},
\end{equation}
and each factor $\ket{C^{(i)}_{k,0}}$ is efficiently preparable by
Theorem~\ref{theorem:piecewise_poly}. Consequently,
using the theorem to prepare each
$\ket{C^{(i)}_{j,0}}$ and then combining via LCU, the total preparation
scales as $\mathcal Q d$, hence linear in $d$. 
\end{remark}

\begin{table}[t!]
\centering
\label{tab:payoff_formulas_final}
\renewcommand{\arraystretch}{1.2} 
\begin{tabular}{c||c}
\textbf{Call option} &
{\large$\displaystyle
e^{-rT}\,\frac{(b-K)^{3/2}}{\sqrt{3(b-a)}}\sqrt{\frac{2^n-1}{2^n}}
\frac{|\langle C_0 \mid p(T)\rangle|}{|\langle +^{\otimes n}\mid p(T)\rangle|}
$} \\
\hline\hline
\textbf{Put option} &
{\large$\displaystyle
e^{-rT}\,\frac{(K-a)^{3/2}}{\sqrt{3(b-a)}}\sqrt{\frac{2^n-1}{2^n}}
\frac{|\langle P_0 \mid p(T)\rangle|}{|\langle +^{\otimes n}\mid p(T)\rangle|}
$} \\
\end{tabular}
\caption{Final call and put option prices obtained from swap tests.  
Here \(a\) and \(b\) respectively denote the lower and upper integration limits of  
the underlying asset domain. \(\ket{C_0}\) encodes the corresponding payoff function (put/call).}\label{tab:TableIV}
\end{table}

\noindent
The combined effect of Monte-Carlo (shot) noise and discretization error on the option prices satisfies the scaling
\begin{equation}
\frac{\epsilon_{V_0}}{V(0)} \sim 
\mathcal{O}\left(1/\sqrt{N_{\text{shots}}}\right) 
+ \mathcal{O}(\Delta x),
\end{equation}
where \(\Delta x=(b-a)/(2^n-1)\) is the grid spacing and  
\(N_{\text{shots}}\) denotes the total number of independent circuit  
executions (or “shots”) used to estimate each overlap from classical  
bitstring outcomes. Full derivations and detailed analysis are provided in  
Appendix~\ref{sec:appendix_swap}.

\section{Conclusion}

In this work, we have introduced a quantum algorithm to solve option-pricing of vanilla options. In particular, we solve the Kolmogorov forward equation under the local-volatility model by employing the Schrödingerisation technique to map the problem of solving the price dynamics into a Hamiltonian simulation task. In this sense, our framework takes classical input (the current stock price) and provides classical data as output (the corresponding option price), thereby establishing a viable pathway toward exploiting quantum computational resources for financial modeling.

Firstly, by contrasting forward and backward formulations, we have elucidated the inherent duality in option-pricing and shown that the forward Kolmogorov approach provides significant numerical benefits, especially in the context of quantum implementations. In this sense, the main advantage of the forward scheme is the efficient quantum retrieval of the option-pricing. Our subsequent analysis regarding the computational implementation highlights that, although the quantum advantage for pricing a single option may be limited to polynomial improvements, the true potential of the proposed methodology lies in its scalability. In particular, the exponential efficiency gain when addressing high-dimensional problems, such as baskets of options, where the complexity scales linearly in terms of the number of dimensions $d$. This underscores the capacity of quantum algorithms to mitigate the curse of dimensionality that constrains classical  methodologies for option-pricing.

Future research should address extending the present framework to incorporate stochastic volatility dynamics, more general payoffs, and risk management applications constitutes a promising direction. Additionally, the exponentially better scaling in dimension motivates the study of more complex financial derivatives like American and Asian options. 

Ultimately, the results reported herein reinforce the view that quantum computing can provide a transformative tool in computational finance, paving the way toward practical quantum advantage in option-pricing and beyond.\\

\section*{Acknowledgments}
NG acknowledges funding from NSFC grant W2442002. MS and JGC acknowledge
support from HORIZON-CL4-2022-QUANTUM01-SGA project 101113946 OpenSuperQ-Plus100 of
the EU Flagship on Quantum Technologies, the Spanish Ram\'on y Cajal Grant RYC-2020-030503-I,
and the “Generaci\'on de Conocimiento” project Grant No. PID2021-125823NA-I00 funded by MICIU/AEI/10.13039/501100011033, by “ERDF Invest in your Future” and by FEDER EU. We also
acknowledge support from the Basque Government through Grants No. IT1470-22, the Elkartek
project KUBIBIT - kuantikaren berrikuntzarako ibilbide teknologikoak (ELKARTEK25/79), and from the IKUR Strategy under the collaboration agreement between Ikerbasque Foundation and BCAM on behalf of the Department of Education of the Basque
Government. This work has also been partially supported by the Ministry for Digital Transformation
and the Civil Service of the Spanish Government through the QUANTUM ENIA project call – Quantum Spain project, and by the European Union through the Recovery, Transformation and Resilience
Plan – NextGenerationEU within the framework of the Digital Spain 2026 Agenda. NL acknowledges funding from the Science and Technology Commission of Shanghai Municipality (STCSM) grant no. 24LZ1401200 (21JC1402900), NSFC grants No.12471411 and No.12341104, the Shanghai Jiao Tong University 2030 Initiative, the Shanghai Science and Technology Innovation Action Plan (24LZ1401200) and the Fundamental Research Funds for the Central Universities.  

\bibliography{Article.bib}
\newpage
\appendix
\onecolumngrid
\section{Auxiliary wave function $\ket{C_0}$ building}\label{appendix: auxiliary wave function}

We start with the encoding of the wave function
\begin{equation}
    \ket{\phi}=\frac{1}{N_\phi}\sum_{i=0}^{2^n-1}(ai+b)\ket{i},
\end{equation}
and then modify the circuit to achieve the desired quantum state, $\ket{C_0}$. Authors in Ref. \cite{gonzalez2024efficient} provide an efficient way on how to build the quantum circuit than encodes $\ket{\phi}$ using the Walsh-Hadamard transform or alternatively its MPS representation.\\




Our next step is to expand the circuit in figure to prepare $\ket{C_0}=\frac{1}{N_f}\sum_{i=\kappa}^{2^n-1}(x_i-K)\ket{i}$. For this purpose we introduce one flag ancillary qubit which keeps the information whether $i\geq\kappa$. The initialization of this state is depicted in Fig.~\ref{fig:K_greater_or_not}. The complexity of this procedure is of $12n-4$ C-NOTs and $n-1$ ancillary qubits.
\begin{eqnarray}
    \ket{\phi}\rightarrow \ket{\zeta}=\frac{1}{N_\phi}\left[\sum_{i=0}^{\kappa-1}(ai+b)\ket{i}\ket{0}+\sum_{i=\kappa}^{2^n-1}(ai+b)\ket{i}\ket{1}\right].
\end{eqnarray}

\begin{figure}[h!]
    \includegraphics[width=1\textwidth]{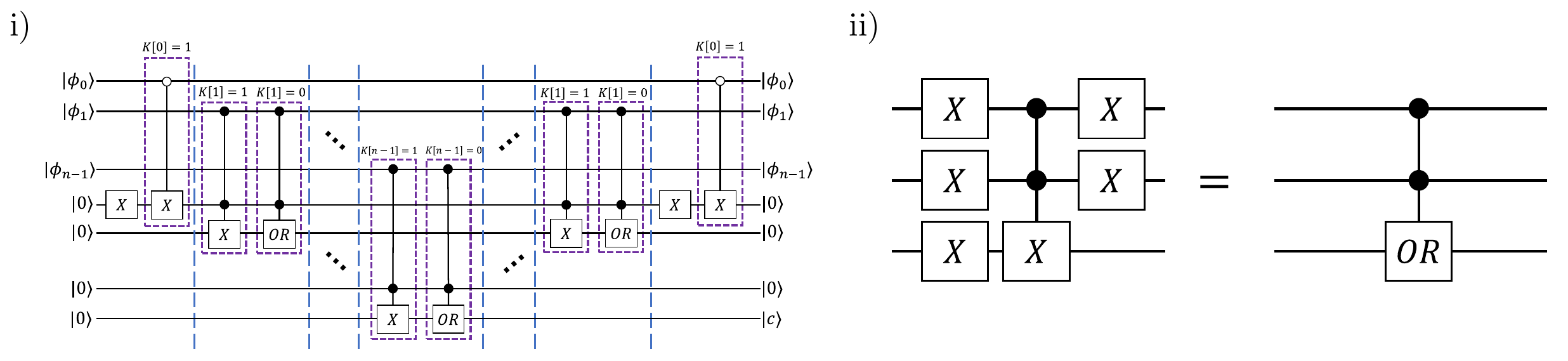}
    \caption{i) The general scheme of building the state $\ket{\zeta}$ from $\ket{\phi}$. The whole circuit can be understood as classical comparator as it compare the state (computational basis) in the upper register with $K$ setting the last qubit in the state $\ket{c}=\ket{1}$ if $\phi\geq K$. The quantum circuit uses $n-1$ ancillas setting them back to zero state $\ket{0}$. A classical array $K[i]$ holds binary representation of $K$. The value of $K[i]$ determines which gate is applied. ii) The representation of $OR$ gate.}
    \label{fig:K_greater_or_not}
\end{figure}

Finally, measuring the last qubit of $\ket{\zeta}$ in $\ket{1}$ give us the desired $\ket{C_0}$ with probability $\mathcal{O}(1)$. This estimation yields from $\kappa \approx 2^{n-1}$ which correspond to the assumption that we introduced a finite-partitioning grid in the region of interest.

\section{Comparison with classical numerical methods}\label{appendix: comparison with classical}
After semi-discretizing Eq.~\eqref{eq:forward_equation_general_form} in space, the PDE is reduced to the linear ODE system
\begin{equation}
  \frac{d u}{d t} = A(t)\,u,\qquad u(0)=u_0,
\end{equation}
where $A(t)\in\mathbb{R}^{N\times N}$, $N=2^{n}$ is the number of spatial grid points, and $A(t)$ has row sparsity $\mathfrak{s}$ (number of nonzero elements per row), where we omitted the inhomogeneous part for simplicity. Two standard classical routes are then employed.

\paragraph{(i) Time discretization (finite-difference time stepping).}
Classical finite-difference schemes—forward/backward Euler and Crank–Nicolson~\cite{leveque2007finite}—discretize both space and time and iteratively update the solution via banded linear algebra. For the implicit backward Euler scheme, each time step requires solving a linear system
\begin{equation}
  A'(t_k)\,u(t_{k-1}) \;=\; b(t_k), \qquad \text{e.g.,}\quad A'(t_k)=I-\Delta t\,A(t_k),
\end{equation}
with $b(t_k)$ incorporating the previous iterate and any source/boundary contributions. With spatial spacing $\Delta x$ and time step $\Delta t$, the total error obeys
\begin{equation}
  \epsilon_\text{fd}\sim \mathcal{O}(\Delta t)+\mathcal{O}(\Delta x^{\,g}),
\end{equation}
where $g$ is the formal order of the spatial stencil (e.g., $g=2$ for second-order central differences; see Table~\ref{table:central-difference-schemes} and~\cite{leveque2007finite}). For parabolic problems, accuracy/stability considerations imply $\Delta t\sim (\Delta x)^2$, so simulating total time $T$ requires $T/\Delta t\sim N^2$ steps. Each step costs $\mathcal{O}(\mathfrak{s}N)$ FLOPs using tridiagonal-/narrow-band solvers, yielding the total complexity
\[
  \mathcal{O}(\mathfrak{s}\,T\,N^{3}) ,
\]
with memory footprint $\mathcal{O}(\mathfrak{s}N)$.

\paragraph{(ii) Exponential integration (matrix–exponential action).}
Another efficient classical approach is the exponential–integrator (matrix–exponential action) method~\cite{doi:10.1137/100788860}. The exact evolution is the time–ordered exponential
\begin{equation}
  u(0) \;=\; \mathcal{T}\exp\!\Big(\int_{T}^{0} A(t)\,dt\Big)\,u(T).
\end{equation}
The idea of the method is to approximate the time–ordered exponential by a product of short–time propagators. Partition $[0,T]$ into $N_t$ steps with $t_k=k\Delta t$ and $\Delta t=T/N_t$; for each subinterval,
\begin{equation}
  e^{\int_{t_k}^{t_{k-1}} A(t)\,dt}\;\approx\; I - A(t_k)\,\Delta t \;+\; \mathcal{O}\big(\lVert A\rVert_{\max}^{2}\,\Delta t^{2}\big),
\end{equation}
and composition yields
\begin{equation}
  u(0)\;\approx\;\big(I - A(t_1)\Delta t\big)\cdots\big(I - A(t_{N_t})\Delta t\big)\,u(T)
  \;=\; \big(I - A(\Delta t)\Delta t\big)\cdots\big(I - A(T)\Delta t\big)\,u(T).
\end{equation}
In the present setting (Eq.~\eqref{eq:forward_equation_general_form}),
\begin{equation}
  \lVert A\rVert_{\max}\;\sim\;\sigma_{\max}^{2}\,\max(\abs{a},\abs{b})^{2}\,2^{2n},
\end{equation}
which is directly analogous to the Hamiltonian bound $\lVert H\rVert_{\max}$ appearing in Theorem~\ref{theorem:QuantumHamiltonianSimulation}. Enforcing the usual first–order condition $\lVert A\rVert_{\max}\Delta t=\Theta(1)$ gives $N_t=\Theta(\lVert A\rVert_{\max}T)$ and total FLOP complexity
\begin{equation}
  \mathcal{O}\big(\mathfrak{s}\,N\,N_t\big)\;=\;\mathcal{O}\big(\lVert A\rVert_{\max}\,T\,\mathfrak{s}N\big)
  \;\sim\;\mathcal{O}(\mathfrak{s}T\,N^{3}),\nonumber
\end{equation}
with memory $\mathcal{O}(\mathfrak{s}N)$. Both classical approaches therefore exhibit the same asymptotic resources in this setting: memory $\mathcal{O}(\mathfrak{s}N)$ and time $\mathcal{O}(T N^{3})$ (simplified as $\mathfrak{s}=3$ for PDE~(\ref{eq:forward_equation_general_form})).

By contrast, the quantum algorithm described here achieves a total complexity
that scales as $\mathcal{O}(N^2polylog(N))$, i.e., quadratically in the number of
grid points, rather than cubically as in the classical cases~\cite{leveque2007finite, doi:10.1137/100788860}.
This demonstrates a clear polynomial quantum advantage for large-scale
PDE problems relevant to option-pricing.

In classical computational approaches, adding extra dimensions—such as those needed to 
account for non-unitarity, time-dependence rapidly leads 
to an exponential increase in computational cost, a phenomenon known as the curse of 
dimensionality. This exponential barrier is a fundamental limitation of classical 
algorithms. In contrast, quantum computation fundamentally changes this scaling: in the 
quantum setting, the introduction of additional registers or auxiliary variables does not 
lead to exponential overhead. In our method, the use of the Schrödingerisation procedure 
and the auxiliary (clock) dimension is essential; these ingredients are required for 
simulating non-unitary and time-dependent quantum dynamics. Importantly, as demonstrated 
in Theorem~\ref{theorem:QuantumHamiltonianSimulation}, the contribution of these extra 
dimensions to the overall quantum complexity is only linear, preserving the main scaling 
of the algorithm. This result indicates that quantum algorithms can efficiently handle 
otherwise classically intractable, high-dimensional financial problems, including basket 
options and correlated multi-asset derivatives. Thus, quantum algorithms open the 
possibility of tackling problems that are beyond reach for classical methods.

\section{Computing option price via the swap test}
\label{sec:appendix_swap}

Using results from Theorem~\ref{theorem: Random walk} and Theorem~\ref{theorem:QuantumHamiltonianSimulation},  
we assume we have prepared the probability amplitude state \(|p(T)\rangle\), representing the underlying asset's  
distribution at maturity. Our goal is to compute the option price \(V\) for the call option  
(the put case will be presented later).

We begin by defining the normalized payoff and probability states as
\begin{eqnarray}
|C_0\rangle &=& \frac{1}{N_f}\sum_{i=0}^{2^n-1} f(x_i)\,|i\rangle, 
\qquad 
|p(T)\rangle \;=\; \frac{1}{N_p}\sum_{i=0}^{2^n-1} p(x_i)\,|i\rangle,
\label{eq:states}
\end{eqnarray}
where \(N_f = \sqrt{\sum_i f(x_i)^2}\) and \(N_p = \sqrt{\sum_i p(x_i)^2}\) ensure normalization. 
For the call payoff \(f(x)=\max(0,x-K)\), the discounted option price is first expressed as  
\begin{eqnarray}
V = e^{-rT} \int_a^b f(x) p(x)\, dx 
\;\approx\; e^{-rT} \sum_{i=0}^{2^n-1} f(x_i) p(x_i)\,\Delta x
\;=\; e^{-rT}\, N_f N_p\, |\langle C_0 \mid p(T)\rangle|\,\Delta x,
\label{eq:V_integral}
\end{eqnarray}
where \(e^{-rT}\) is the (constant-rate) risk-free discount factor, and \(N_f\) is analytically computable:
\begin{eqnarray}
N_f &=& \Bigg(\sum_{i=0}^{2^n-1} f(x_i)^2\Bigg)^{1/2}
\;\approx\; \left(\frac{1}{\Delta x}\int_a^b f(x)^2\,dx\right)^{1/2}
\nonumber\\
&=& \frac{1}{\sqrt{3\,\Delta x}}\,
\left\{
\begin{array}{ll}
(b-K)^{\tfrac{3}{2}}, & \text{if } f(x)=(x-K)_+ \quad \text{(call)},\\[6pt]
(K-a)^{\tfrac{3}{2}}, & \text{if } f(x)=(K-x)_+ \quad \text{(put)}.
\end{array}
\right.
\label{eq:Nf_call_put}
\end{eqnarray}
The overlap \(|\langle C_0 \mid p(T) \rangle|^2\) is obtained from the swap test (Theorem~\ref{theorem:swap test}). 
To compute \(N_p\) we use a \emph{second} swap test between \(|p(T)\rangle\) and the uniform state \(|+\rangle^{\otimes n}\).  
Starting from
\begin{eqnarray}
|\langle +^{\otimes n}\mid p(T)\rangle| 
= \frac{\sum_i p(x_i)}{\sqrt{2^n}\,N_p}
= \frac{\sum_i p(x_i)\Delta x}{\sqrt{2^n}\,N_p\,\Delta x}
\;\approx\; \frac{1}{\sqrt{2^n}\,N_p\,\Delta x},
\label{eq:fidelity_start}
\end{eqnarray}
where we used that \(p(x)\) is a probability distribution on \([a,b]\), i.e.
\begin{eqnarray}
\int_{a}^{b} p(x)\,dx \;\approx\; 1 
\quad \Longrightarrow \quad 
\sum_i p(x_i)\,\Delta x \;\approx\; 1.
\end{eqnarray}
Substituting into Eq.~\eqref{eq:fidelity_start}, we obtain
\begin{eqnarray}
N_p = \frac{1}{\Delta x\,\sqrt{2^n}\,|\langle +^{\otimes n}\mid p(T)\rangle|}.
\label{eq:Np_formula}
\end{eqnarray}

Finally, setting \(\Delta x = \frac{b-a}{2^n-1}\),  
we obtain the call option price as
\begin{eqnarray}
V_{\text{call}} = e^{-rT}\;\frac{(b-K)^{3/2}}
{\sqrt{3(b-a)}}\sqrt{\frac{2^n-1}{2^n}}
\frac{|\langle C_0 \mid p(T)\rangle|}
{|\langle +^{\otimes n}\mid p(T)\rangle|}.
\label{eq:V_call}
\end{eqnarray}
For the put option, define \( |P_0\rangle = \frac{1}{N_f^{(\text{put})}}\sum_i (K-x_i)_+\,|i\rangle \) with 
\(N_f^{(\text{put})}\) as in Eq.~\eqref{eq:Nf_call_put}. Then
\begin{eqnarray}
V_{\text{put}} = e^{-rT}\;\frac{(K-a)^{3/2}}
{\sqrt{3(b-a)}}\sqrt{\frac{2^n-1}{2^n}}
\frac{|\langle P_0 \mid p(T)\rangle|}
{|\langle +^{\otimes n}\mid p(T)\rangle|}.
\label{eq:V_put}
\end{eqnarray}

\subsection{Shot noise}
When implementing the swap tests in practice, we  
run a finite number of quantum circuits and the statistics of the ancilla estimate the overlaps 
\(|\langle +^{\otimes n}\mid p(T)\rangle|^2\) and \(|\langle C_0\mid p(T)\rangle|^2\). 
This sampling uncertainty (shot noise) introduces statistical fluctuations depending on the number of samples (shots).  
From Theorem~\ref{theorem:swap test}, together with Eqs.~\eqref{eq:V_call}–\eqref{eq:V_put}, error propagation yields  
\begin{eqnarray}
\frac{\tilde{\epsilon}_{V}}{V} \;\approx\; \frac12 \sqrt{\frac{\epsilon_{F_1}^2}{F_1^2}  
+ \frac{\epsilon_{F_2}^2}{F_2^2}} \;\sim\;\mathcal{O}\!\ \bigl(N_{\text{shots}}^{-1/2}\bigr),
\label{eq:error_V}
\end{eqnarray}
where \(F_1 = |\langle +^{\otimes n}\mid p(T)\rangle|^2\) and \(F_2=|\langle C_0\mid p(T)\rangle|^2\)  
are estimated from independent swap-test samples, each with 
\(\epsilon_{F_i}\sim \mathcal{O}(N_{\text{shots}}^{-1/2})\). 
Thus the relative sampling error scales as \(\mathcal{O}(N_{\text{shots}}^{-1/2})\) for both  
call and put options.

\subsection{Discretization error}
Replacing the integrals in  
Eqs.~\eqref{eq:V_integral} and \eqref{eq:fidelity_start} by discrete sums  
introduces a deterministic bias from the finite grid spacing  
\(\Delta x = (b-a)/(2^n-1)\). For the Riemann-sum formulas both the price integral and the probability  
normalization contribute first-order errors \(O(\Delta x)=O(2^{-n})\). Balancing  
this with the shot-noise fluctuations in Eq.~\eqref{eq:error_V}, the total  
relative error on the option price caused by information retrieval satisfies  
\begin{eqnarray}
\frac{\epsilon_V}{V} \;\sim\; O(2^{-n}) \;+\; O(N_{\text{shots}}^{-1/2}).
\label{eq:total_error}
\end{eqnarray}

\end{document}